\documentclass[notitlepage]{article}
\usepackage[titletoc,title]{appendix}
\usepackage{longtable}
\usepackage{fullpage}
\usepackage{times}  
\usepackage{amsmath}
\usepackage{amssymb}
\usepackage{amsthm}
\usepackage{float}
\usepackage{bbm}
\usepackage{bm}
\restylefloat{table}
\restylefloat{figure}
\usepackage{rotating}
\usepackage{tablefootnote}
\usepackage[T1]{fontenc}
\usepackage{color}
\usepackage{breqn}
\usepackage{wrapfig}
\usepackage{amsfonts}
\usepackage{booktabs}
\usepackage{hyperref}
\usepackage[authoryear]{natbib}
\usepackage[space]{grffile} 
\usepackage{authblk}

\usepackage{multirow}
\usepackage[shortlabels]{enumitem}
\usepackage{subfigure}
\graphicspath{{./}}

\numberwithin{equation}{section}
\newtheorem{definition}{Definition}[section]

\newtheorem{theorem}{Theorem}[section]
\newtheorem{corollary}{Corollary}[section]

\DeclareMathOperator{\R}{\mathbb{R}}

\DeclareMathOperator{\mE}{\mathcal{E}}
\DeclareMathOperator{\mF}{\mathcal{F}}
\DeclareMathOperator{\mG}{\mathcal{G}}

\DeclareMathOperator{\mL}{\mathcal{L}}
\DeclareMathOperator{\mM}{\mathcal{M}}
\DeclareMathOperator{\mQ}{\mathcal{Q}}
\DeclareMathOperator{\mR}{\mathcal{R}}

\DeclareMathOperator{\mV}{\mathcal{V}}
\DeclareMathOperator{\mW}{\mathcal{W}}
\DeclareMathOperator{\Oo}{\mathcal{O}}
\DeclareMathOperator{\ad}{\text{ad}}
\hypersetup{colorlinks=false, pdfborder={0 0 0}}

\begin{document}

\title{Analytic RFR Option Pricing with Smile and Skew}
\author[*]{Colin Turfus}
\author[*,$\dag$]{Aurelio Romero-Bermudez}
\affil[*]{Deutsche Bank}
\affil[$\dag$]{ABN Amro}

\date{Initial Version: December 22, 2022 \\[.1cm] Current Version: \today }
\maketitle

\begin{abstract}
We extend the short rate model of \cite{SmilesWithoutTears} to facilitate accurate arbitrage-free analytic pricing of SOFR, SONIA or ESTR caplets, i.e.~options on backward-looking compounded rates payments, in a manner consistent with the smile and skew levels observed in the market. These caplet pricing formulae and corresponding LIBOR or term-rate caplet results are translated into effective variance (implied vol) formulae, which are seen to be of a particularly simple form. They show that the model is essentially equivalent to imposing on a Hull-White model an effective variance which is a quadratic function of the moneyness parameter (rather than a constant) for any given maturity. Results are also illustrated graphically.
\end{abstract}

\section{Introduction} \label{Introduction}

The short rate model of \cite{SmilesWithoutTears} allows analytic pricing of caplets based on term rates such as LIBOR with smile and skew. We extend this model here to support analytic pricing of options on \emph{compounding risk-free rates} (RFR) such as SOFR, SONIA or ESTR, in a manner consistent with the smile and skew levels observed in the market. Our approach is based on the observation that for small perturbations from equilibrium under this model the short rate is essentially a linear function of a (mean-reverting) Gaussian variable $y$. This makes the situation similar to that captured by the model of \cite{H-W}. An exact analytic pricing kernel for this model was derived by \cite{Steenkiste,Exact_H-W} and the extension of this kernel to address compounding rates by \cite{PerturbationMethods,Compounded_Rates}.

For a review of related work on options on compounded rates, the reader is referred to \cite{Compounded_Rates}, which follows the approach of \cite{Steenkiste} in introducing the integrated short rate as an additional state variable.  Option prices are in this way derived analogous to the formulae of \cite{Henrard_H-W}, including taking account of delayed payment of the option premium. This work was extended to incorporate counterparty credit risk by \cite{RiskyCompoundedRatesRisk} through a highly accurate asymptotic approximation.

In relation to analytic option pricing with smile and skew, the reader is referred to \cite{SmilesWithoutTears,WhatShortRateModel}. Their approach, which we extend below, is based on modelling the short rate as a $\sinh$ function of a Gaussian Ornstein-Uhlenbeck process, which they argue is equivalent to assuming a quadratic dependence of the local volatility on the short rate. They make the point that little work has been done in this space, albeit that the Beta Blend model proposed by \cite{BetaBlendGreen} interpolating between the model of \cite{H-W} and that of \cite{B-K} can be viewed as allowing a term structure of skew to be specified. Also analytic LIBOR option and swaption prices incorporating smile and skew are available under the SABR model, as discussed by \cite{ModernSABRAnalytics}. But this approach allows only one rate to be modelled at a time and does not constitute a short rate model capable of capturing multiple rates within a unified framework.

We single out for special mention here the work of \cite{Antonov}. Like us they started out by considering Gaussian variables driven by Ornstein-Uhlenbeck dynamics and took the short rate to be a function thereof. One of the two cases they considered%
\footnote{The other case addressed was the well-known lognormal rates model of \cite{B-K} (see further \S \ref{sec:Description} below).}
was a quotient of affine functions of exponentials of the Gaussian, under which assumption the short rate was constrained to lie between an upper and a lower limit. They did not like us attempt to capture volatility smile and skew explicitly. Nor did they consider a more general formulation involving an affine function of \emph{two} exponentials with different exponents as we do here. Also, the (first-order) perturbation expansion they used to obtain analytic option prices was based on a small lognormal-volatility assumption, which limits the accuracy which can be obtained, particularly under low rates where lognormal volatilities can be quite high (in excess of 50\%). We expect our expansion based on the smallness of the rates themselves to be more accurate, especially in circumstances of low rates.

We look to obtain analytic expressions for prices of SOFR options which take account of smile and skew by a method similar to that employed by \cite{SmilesWithoutTears} to price LIBOR caplets based on an extension of the Bachelier pricing kernel (for a Gaussian underlying). We do this by extending instead the Hull-White compounded rates kernel of \cite{Compounded_Rates}. The challenge is that, as well as observing that the representations of the interest rate are similar (linear functions of a Gaussian variable) with and without smile or skew, we need to find a means of taking advantage of that observation systematically to obtain an explicit mathematical representation of the perturbation which must be applied to the known (Hull-White) analytic kernel. We do this using the operator expansion method expounded by \cite{PerturbationMethods}.

We start off in \S \ref{sec:Description} by setting out the pricing equation for our model. The associated pricing kernel is presented in \S\ref{sec:PricingKernel} as a perturbation expansion based on a low-rates assumption, with a summary proof provided in Appendix \ref{sec:PricingKernelProof}. The calibration of the model to the term structure of forward rates and volatility (with skew and smile) is considered in \S \ref{sec:Calibration}, where bond price formulae are also set out. The further application of the kernel to the pricing of forward rates and SOFR/ESTR caplets is considered in \S \ref{sec:Applications}, with a proof of the latter in Appendix \ref{sec:CapletPricingProof}. A LIBOR caplet pricing formula is also deduced providing more accuracy than that of \cite{SmilesWithoutTears}. The caplet prices are also give a convenient expression as implied volatilities, which simplifies calculation and helps ensure consistency of behaviour far from the money. Numerical calculations of caplet prices and implied volatility surfaces based on the formulae deduced are illustrated in \S \ref{sec:Numerical}. Conclusions are set out in \S \ref{sec:Conclusions}.

\section{Model Description}\label{sec:Description}

We consider the short rate model of \cite{SmilesWithoutTears} which incorporates smile and skew into the Hull-White model through a variable transformation. We shall find it convenient to work with the reduced variable $y_t$  defined through the following Ornstein-Uhlenbeck process:
\begin{equation}
dy_t = -\alpha(t) y_t dt + \sigma(t)\,dW_t,\label{eq:dy_t}
\end{equation}
with $W_t$ a Wiener process for $t\ge0$. This auxiliary variable is related to the short rate by $r_t =r(y_t,t)$ where
\begin{align}
	r(y,t)&:=\overline{r}(t) + R^*(t) + \frac{\sinh \gamma(t)(y+y^*(t))}{\gamma(t)}.\label{eq:r(y,t)}
\end{align}
Here $\overline{r},y^* :\mathbb{R}^+\to\mathbb{R}$ are the instantaneous forward rate and a skewness function, respectively, and $\sigma,\alpha, \gamma, R^*:\mathbb{R}^+\to\mathbb{R}^+$ are functions representing volatility, mean reversion rate, a smile factor and a convexity adjustment factor, respectively, all assumed to be piecewise continuous and bounded.
We further assume that $y_0=0$, with $t=0$ the ``as of'' date for which the model is calibrated. The function $R^*(t)$ is determined by calibration to the forward curve but will tend in the zero volatility limit to zero.%
\footnote{Strictly speaking, the function $y^*(\cdot)$ is the drift adjustment parameter which must be chosen to satisfy the no-arbitrage condition for a given choice of $R^*(t)$, so the latter should really be thought of as controlling the skew. However, the calculation of $R^*(t)$ as a power series of nested integrals involving $y^*(\cdot)$ is much more convenient than attempting to perform the equivalent process in reverse, so in practice we consider $y^*(t)$ to be inferred directly from the no-arbitrage condition, constructing the implied $R^*(t)$ using the mathematical techniques expounded below.}
The formal no-arbitrage constraint which determines this function is as follows
\begin{equation}
E\left[e^{-\int_{0}^tr_s\,ds}\right]=D(0,t)\label{eq:No_arbitrage}
\end{equation}
under the martingale measure for $0 < t\le T_m$, where $T_m$ is the longest maturity date for which the model is calibrated, and
\begin{equation}
D(t_1,t_2):=e^{-\int_{t_1}^{t_2}\overline{r}(s)\,ds}
\end{equation}
is the $t_1$-forward price of the $t_2$-maturity zero coupon bond. This is to be contrasted with the \cite{H-W} (H-W) model which is obtained by taking the limit as $\gamma(t)\to0$ in  \eqref{eq:r(y,t)}:
$$r(y,t)=R^*_{\text{H-W}}(t)+y$$
and the \cite{B-K} (B-K) model which is obtained by choosing
$$r(y,t) = R^*_{\text{B-K}}(t)e^y.$$
Notably our model reverts to the first of these in the limit as $\gamma(t)\to0$.

Since we are interested in pricing options on daily-compounded rates payments, we introduce a new variable $z_t$ representing the integral of the short rate, making use of the fact that daily compounding can be to a good approximation captured by assuming continuous compounding. We shall make the choice
\begin{equation}\label{eq:z_t}
	z_t:= \int_0^t (r_s - \overline{r}(s))ds.
\end{equation}

We consider the (stochastic) time-$t$ price of a European-style security which pays a cash amount $P(y_T,z_T)$ at maturity $T$, denoting this by  $f(y_t,z_t,t)$. We note in particular that the price of a $T$-maturity zero coupon bond is obtained by taking $P(y,z) = 1$. We will look to derive the functional form of $f(\cdot,\cdot,\cdot)$ implied by our model in this case and in the process to determine the conditions on $R^*(t)$ necessary to satisfy \eqref{eq:No_arbitrage}.

From the Feynman-Kac theorem we infer that $f(y,z,t)$ emerges as the solution to the following backward Kolmogorov diffusion equation:
\begin{equation}\label{eq:fPDE}
\frac{\partial f}{\partial t}-\alpha(t)y\frac{\partial f}{\partial  y} +\left(r(y,t)- \overline{r}(t) \right)\left(\frac{\partial}{\partial z}-1\right)f +\frac12\sigma^2(t) \frac{\partial^2 f}{\partial y^2} - \overline{r}(t)f = 0,
\end{equation}
with $r(y,t)$ given by \eqref{eq:r(y,t)}, subject to the termiinal condition $f(y,z,T)=P(y,z)$.

\section{Pricing Kernel}\label{sec:PricingKernel}

We define at this point the following notation for future reference:
\begin{align}
	\phi_r(t,v)&:=e^{-\int_t^v \alpha(u)du},\label{eq:phi_r}\\
	\Sigma_{rr}(t,v)&:=\int_t^v\phi_r^2(u,v)\sigma^2(u)du,\label{eq:Sigma_rr}\\
	\psi_r(t,v)&:= e^{\frac12\gamma^2(v)\Sigma_{rr}(t,v)},\label{psi_r}\\
	\Sigma_{rz}(t,v)&:=\int_t^v \psi_r(t,u) \phi_r(u,v)\Sigma_{rr}(t,u)du, \label{eq:Sigma_rz}\\
	\Sigma_{zz}(t,v)&:=2\int_t^v \psi_r(t,u) \Sigma_{rz}(t,u)du, \label{eq:Sigma_zz}\\
	\bm{\Sigma}(t,v)&:= \left(\begin{array}{c c}
		\Sigma_{rr}(t,v) &\Sigma_{rz}(t,v)\\
		\Sigma_{rz}(t,v)&\Sigma_{zz}(t,v)
		\end{array}\right),\label{eq:Sigma}\\
	B^*(t,v)&:=\int_t^v \psi_r(t,u)\phi_r(t,u)du,\label{eq:B*}\\
	B^+(t,t_1,v)&:=\frac{B^*(t,v)-B^*(t,t_1)}{\phi_r(t,t_1)}\cr
	&=\int_{t_1}^v \psi_r(t,u)\phi_r(t_1,u)du,\label{eq:B+}\\
	\mu^*(y,t,v)&:= B^*(t,v)(y +\Sigma_{rz}(0,t)) +\textstyle\frac12B^{*2}(t,v)\Sigma_{rr}(0,t).\label{eq:mu*}
\end{align}
We observe here that $\phi_r(t,v)$ represents the impact of mean reversion and $\psi_r(t,v)$ that of the counteracting dispersion from the mean induced by smile.

We look to obtain the pricing kernel or Green's function for \eqref{eq:fPDE}, defined as a (generalised) function $G(y,z,t;\eta,\zeta,T)$ such that the solution subject to the specified terminal condition can be obtained from the following convolution expression:
\begin{equation}\label{eq:G-def}
	f(y,z,t)=\iint_{\R^2}G(y,z,t;\eta,\zeta,T)P(\eta,\zeta)\,d\eta \,d\zeta.
\end{equation}
In the absence of exact analytic solutions, we seek the pricing kernel as a perturbation expansion in the limit that the deviations of the short rate from the forward rate curve are small under some suitable norm. To that end let us define the small parameter
\begin{equation}\label{eq:epsilon}
	\epsilon:= \left\|\gamma^2(\cdot)\Sigma_{rr}(\cdot,\cdot)\right\|
\end{equation}
and seek an asymptotic expansion of the pricing kernel in powers of $\epsilon$. To allow of a skew adjustment that contributes at the same order of approximation as the smile adjustment we shall make the assumption that $\|\gamma(\cdot)y^*(\cdot)\| = \Oo(\sqrt{\epsilon}).$%
\footnote{We have avoided being explicit about the range of $t$ over which the norms are assessed. A typical choice would be a weighted integral over $[0,T]$ for some finite $T$.}
Let us further express $R^*(\cdot)$ asymptotically as
\begin{equation}\label{eq:R*_expansion}
	R^*(t)=\sum_{n=1}^{\infty} R_n^*(t),
\end{equation}
with $R_n^*(\cdot)=\Oo(\epsilon^n)$.

To express our result conveniently, we introduce the functions
\begin{align}\label{eq:R_1^+-}
	R_1^{\pm}(y,t,t_1,v)&:=\psi_r(t,t_1)\frac{e^{\pm\gamma(t_1)(\phi_r(t,t_1)y +y^*(t_1)-B^+(t,t_1,v)\Sigma_{rr}(t,t_1)-\Sigma_{rz}(t,t_1))}} {2\gamma(t_1)},
\end{align}
and shift operators $\mM^\pm(t,t_1)$ defined for $t\le t_1$ by
\begin{align}
	\mM^{\pm}(t,t_1)f(y,z,\ldots)&:=f\left(y\pm\gamma(t_1)\Delta y(t,t_1), z \pm\gamma(t_1)\Delta z(t,t_1), \ldots\right) \label{eq:M+-}\\
	\Delta y(t,t_1)&:=\frac{\Sigma_{rr}(t,t_1)}{\phi_r(t,t_1)}, \label{eq:Delta_y}\\
	\Delta z(t,t_1)&:=\Sigma_{rz}(t,t_1) -B^*(t,t_1)\frac{\Sigma_{rr}(t,t_1)}{\phi_r(t,t_1)}.\label{eq:Delta_z}
\end{align}
Using this notation we define also the following operators:
\begin{align}
	\mR_1(t,t_1,v)&:=R_1^+(y,t,t_1,v)\mM^+(t,t_1) -R_1^-(y,t,t_1,v)\mM^-(t,t_1), \label{eq:R}\\
	\overline{\mR}_1(t,t_1,v)&:=R_1^+(y,t,t_1,v)\mM^+(t,t_1) +R_1^-(y,t,t_1,v) \mM^-(t,t_1), \label{eq:R__bar}\\
	\mG_1(t,t_1,v)&:=\mR_1(t,t_1,v) +R_1^*(t_1) -e^{\frac12\gamma^2(t_1)\Sigma_{rr}(t,t_1)} \bigg(\phi_r(t,t_1)(y +\Sigma_{rz}(0,t) +B^*(t,t_1)\Sigma_{rr}(0,t))\cr
	&\hspace{200pt}-B^+(t,t_1,v)\Sigma_{rr}(t,t_1)\bigg).\label{eq:mG_1}
\end{align}
Using the operator expansion method of \cite{PerturbationMethods} in the same manner as \cite{SmilesWithoutTears} we obtain the following result.

\begin{theorem}[Pricing Kernel]\label{PricingKernel}
	Making use of the above-defined notation, the pricing kernel for \eqref{eq:fPDE} can be written asymptotically in the limit as $\epsilon\to0$ as
\begin{equation}\label{eq:G_compounded}
	G(y,z,t;\eta,\zeta,v)=D(t,v)e^{-\mu^*(y,t,v)}\sum_{n=0}^\infty G_n(y,z,t;\eta,\zeta,v)
\end{equation}
with the $G_n(\cdot)=\Oo(\epsilon^n)$ and in particular
\begin{align}
	G_0(y,z,t;\eta,\zeta,v)&=N_2\left(\eta +\Sigma_{rz}(t,v) -\phi_r(t,v)y,\zeta -\mu^*(y,t,v) +\textstyle\frac12\Sigma_{zz}(t,v) -z; \bm{\Sigma}(t,v)\right), \label{eq:G_0}\\
	G_1(y,z,t;\eta,\zeta,v)&=\left(\int_t^v \mG_1(t,t_1,v)dt_1 -\mQ(t,v) \right) \left(\frac{\partial}{\partial z} -1\right) G_0(y,z,t;\eta,\zeta,v), \label{eq:G_1}\\
	G_2(y,z,t;\eta,\zeta,v)&=\Bigg(\int_t^v\int_t^{t_2} \mG_1(t,t_1,v) \mG_1(t,t_2,v)dt_1 dt_2 \left(\frac{\partial}{\partial z} -1\right)^2 \cr
	&\hspace{22pt}-\int_t^v\mG_1(t,t_1,v)dt_1 \mQ(t,v) \left(\frac{\partial}{\partial z} -1\right)^2\cr
	&\hspace{22pt}+\frac12\left(\mQ^2(t,v) -\int_t^v \Sigma_{rz}(t,t_1)(\gamma(t_1)\overline{\mR}_1(t,t_1,v)-1) dt_1\right) \left(\frac{\partial}{\partial z} -1\right)^2\cr
	&\hspace{22pt}+\int_t^v R_2^*(t_1)dt_1\bigg(\frac{\partial}{\partial z} -1\bigg)\Bigg) G_0(y,z,t;\eta,\zeta,v),\quad\label{eq:G_2}
\end{align}
with
\begin{align}\label{eq:Q}
	\mQ(t,v)&:=\frac{\Sigma_{rz}(t,v)}{\phi_r(t,v)}\left(\frac{\partial}{\partial y}-B^*(t,v)\frac{\partial}{\partial z}\right) +\Sigma_{zz}(t,v) \frac{\partial}{\partial z}
\end{align}
and $N_2(\cdot,\cdot;\bm{\Sigma})$ a bivariate Gaussian probability density function with covariance $\bm{\Sigma}$.
\end{theorem}
\begin{proof}
A proof of this result is given in Appendix \ref{sec:PricingKernelProof} below. Note in \eqref{eq:G_2} that the shift operator in $\mG_1(t,t_1,v)$ acts on the $y$-variables in $\mG_1(t,t_2,v)$.
\end{proof}
What has effectively been achieved here is that representations of the $\Oo(1)$ drift and diffusion associated with $z$ and the $\Oo(\epsilon^{1/2})$ Hull-White stochastic discounting, as captured by $\mu^*(y,t,v) +\mQ(t,v)$, are subtracted out of the time-ordered exponential operator $\mW(t,u)$ in \eqref{eq:W} and incorporated directly into the modified Gaussian kernel $G_0(\cdot;\cdot)$ on which it acts. In this way a zeroth order operator is replaced by the first order operator $\mW_1(t,u)$ in \eqref{eq:W_1}. Only by use of this device does the asymptotic expansion of the (modified) operator become possible.

It will be observed that the leading-order contribution ($n=0$) is in its functional form identical to the Hull-White kernel of \cite{Compounded_Rates}, with the caveat that the mean reversion function $\phi_r(t,u)$ is adjusted up (reducing its impact) by inclusion of the smile factor $\psi_r(t,t_1)$ in \eqref{eq:B*}, with no impact from skew. This means that solutions based on our expansion are likely to retain better accuracy even at longer maturities than those of \cite{SmilesWithoutTears}. Note that we could have expressed our results as a perturbation operator acting on the Hull-White kernel but choose, for reasons of convenience of use, to multiply by the stochastic discounting operator in the ansatz \eqref{eq:G_compounded} \emph{after} applying the perturbation operators term-by-term.

\section{Fitting to the Forward Curve}\label{sec:Calibration}

Let us further define for future notational convenience
\begin{align}
	\mF_1(t,T)&:=\int_t^T \mG_1(t,t_1,T)dt_1\cr
	&\phantom{:}=\int_t^T (\mR_1(t,t_1,T)+R_1^*(t_1))dt_1 -\mu^*(y,t,T) +\frac12\Sigma_{zz}(t,T) \label{eq:F_1}\\
	\mF_2(t,T)&:=\int_t^T\left(\int_t^{t_2} \mG_1(t,t_1,T)\mG_1(t,t_2,T)dt_1 -R_2^*(t_2)\right)dt_2. \label{eq:F_2}
\end{align}
Note that $y$-shifts apply in the integral term in \eqref{eq:F_2}, but can in practice be ignored, since they give rise to terms impacting at higher than second order; also that, in deriving \eqref{eq:F_1}, we have made use of the identity \eqref{eq:B^+_identity}.

The calibration of our model to the forward interest rate market is achieved by considering zero coupon bond prices.
\begin{theorem}[Calibration] \label{Calibration}
	Applying \eqref{eq:G_compounded} to a unit payoff at $T$, the $T$-maturity zero coupon bond price associated with the pricing kernel \eqref{eq:G_compounded} is seen to be given asymptotically by
	\begin{align}
		F^T(y,t)\sim D(t,T)e^{-\mu^*(y,t,T)}\left(1-\mF_1(t,T)+\mF_2(t,T)\right).\qquad \label{eq:F^T}
	\end{align}
	Differentiating the calibration condition $F^T(0,0)=D(0,T)$ w.r.t.~$T$, applying this to second order accuracy and setting $T=t$ gives rise to the requirement that the first two terms of \eqref{eq:R*_expansion} are given by
	\begin{align}
		R_1^*(t)&=-R_1^+(0,0,t,t) +R_1^-(0,0,t,t)\cr
		&\quad-\psi_r(0,t)\int_0^t \phi_r(t_1,t)\Sigma_{rr}(0,t_1) (\gamma(t_1)(R_1^+(0,0,t_1,t)+R_1^-(0,0,t_1,t)) -\psi_r(0,t_1))dt_1\cr
		&=-\psi_r(0,t)\left(\frac{\sinh Y^*(t,t)} {\gamma(t)} +\int_0^t \psi_r(0,t_1) \phi_r(t_1,t) \Sigma_{rr}(0,t_1) \left(\cosh 	Y^*(t_1,t) -1\right)dt_1\right),\qquad \label{eq:R_1*}\\
		R_2^*(t)&\sim R_1^*(t) \int_0^t R_1^*(t_1)dt_1,\label{eq:R_2*}
	\end{align}
	where for $t_1\in[0,t]$,
	\begin{equation}\label{eq:Y^*}
			Y^*(t_1,t):=\gamma(t_1) \left(y^*(t_1) -B^+(0,t_1,t)\Sigma_{rr}(0,t_1) -\Sigma_{rz}(0,t_1)\right).
	\end{equation}
\end{theorem}
We observe that, excluding the bracketed expression which provides the main skew-smile adjustment, \eqref{eq:F^T} is essentially a Hull-White representation of the bond price making use of a modified mean reversion factor in place of the usual $\phi_r(t,u)$ as the integrand in \eqref{eq:B*}. A first-order representation which should suffice for many purposes, is obtained by neglecting $\mF_2(t,T)$ in \eqref{eq:F^T}. Provided the leading-order skew-smile correction terms remain small, the approximation should be a good one: the variable $y$ would have to take on improbably large values for this not to be the case. Note that these results are, up to second order, equivalent to those presented in \cite{SmilesWithoutTears}, as might be expected since the model is unchanged here, only the method of building the expansion slightly different.

\section{Applications}\label{sec:Applications}

\subsection{Forward Rates}\label{FwdRate}

An expression for forward rates is easily deduced from our zero coupon bond formula.

\begin{corollary}[Instantaneous Forward Rate]
	The instantaneous forward rate is given by
	\begin{align}\label{eq:f^T}
		f^T(y,t)&=\overline{r}(T)+\psi_r(t,T)\phi_r(t,T)(y +\Sigma_{rz}(0,t) +B^*(t,T) \Sigma_{rr}(0,t))\cr 	&\quad+\frac{D(t,T)e^{-\mu^*(y,t,T)}}{F^T(y,t)}\Bigg(\mG_1(t,T,T) -\int_t^T
		 \mG_1(t,t_1,T)\mG_1(t,T,T)dt_1 +R_2^*(T)\Bigg) +\Oo(\epsilon^{3}),
	\end{align}
	with $F^T(y,t)$ as given by Theorem \ref{Calibration}. The corresponding first order version of this result is obtained by ignoring the second two terms in parentheses and taking a first order representation of $F^T(y,t)$.
\end{corollary}
\begin{proof}
	This result is obtained by noting that the instantaneous forward rate is given by $f^T(y_t,t)$ where
	$$f^T(y,t)= -\frac{\partial}{\partial T}\ln F^T(y,t)$$
and making use of \eqref{eq:F^T}.
\end{proof}

\subsection{Option Pricing}\label{subsec:OptionPricing}

We assume below that, in option pricing, the interest rate that is used as the underlying is the same one that is used in the discounting of cash flows. If, particularly in the case of LIBOR, there is a spread over the risk-free rate, this situation can effectively be managed by calibrating the model to the risk-free rate and subtracting the (assumed deterministic) spread off of the strike (or coupon rate).

\paragraph{Compounded Rate Caplet Pricing}\label{CompoundedCapletPricing}

RFR cap and caplet prices are obtained to leading order from the following theorem.

\begin{theorem}[RFR Caplet Price]\label{CapletPrice}
	Consider a caplet based on the compounded risk-free rate over a payment period $[T_1,T_2]$ and a payoff with strike $K$ at time $T_2$ of
	\begin{align}\label{eq:P_caplet}
		\left[e^{\int_{T_1}^{T_2}r(y_t,t)dt}-1-K\delta(T_1,T_2)\right]^+
		&=\left[D(T_1,T_2)^{-1}e^{z_2-z_1}-\kappa^{-1}\right]^+\cr
		&=:	P_{\text{caplet}}(z_1,z_2)
	\end{align}
	where $z_{T_1}=z_1$, $z_{T_2}=z_2$ and $\kappa=(1+K\delta(T_1,T_2))^{-1}$, with $\delta(\cdot,\cdot)$ providing the relevant year fraction. Using the above-defined notation and defining the critical value of $z_2-z_1$ at which the option comes into the money as
	\begin{equation}
		\Delta z^*=\ln\left(\kappa^{-1}D(T_1,T_2)\right),\label{eq:Delta_z^*}
	\end{equation}
	the caplet PV will be given asymptotically with relative errors $=\Oo(\epsilon^2)$ by
	\begin{align}
		PV_{\text{Caplet}}&\sim D(0,T_1)\left(1-\int_0^{T_1}\tilde{\mG}_1(0,t_1,T_1)  dt_1\right)\Phi(d_1(y,z-z_1, 0))\Bigg|_{y=0,z=z_1}\cr
		&\quad-\kappa^{-1}D(0,T_2)\left(1 -\int_0^{T_2}\tilde{\mG}_1(0,t_1,T_2)e^{-\theta(y,0) H(T_1-t_1)}dt_1 -B^*(T_1,T_2)\Sigma_{rz}(0,T_1)\right)\cr
		&\hspace{290pt}\Phi(d_2(y,z-z_1, 0))\Bigg|_{y=0,z=z_1} \qquad\cr
		&\quad-\kappa^{-1}D(0,T_2)\sqrt{B^{*2}(T_1,T_2)\Sigma_{rr}(0,T_1) +\Sigma_{zz}(T_1,T_2)} N(d_2(0,0,0)), \label{eq:PV_caplet}
	\end{align}
	with $\Phi(\cdot)$ a cumulative normal distribution function, $N(\cdot)$ the corresponding density and $H(\cdot)$ the Heaviside step function, where we define
	\begin{align}
		d_1(y,w,t)&:=\frac{\theta(y,t) +w -\Delta z^* +\frac12(B^{*2}(T_1,T_2) \Sigma_{rr}(t,T_1) +\Sigma_{zz}(T_1,T_2))} {\sqrt{B^{*2}(T_1,T_2)\Sigma_{rr}(t,T_1) +\Sigma_{zz}(T_1,T_2)}}\label{eq:d_1}\\
		d_2(y,w,t)&:=d_1(y,w,t) -\sqrt{B^{*2}(T_1,T_2)\Sigma_{rr}(t,T_1) +\Sigma_{zz}(T_1,T_2)},\label{eq:d_2}\\
		\theta(y,t)&:=B^*(T_1,T_2)(\phi_r(t,T_1)y-\Sigma_{rz}(t,T_1)+\Sigma_{rz}(0,T_1)) +\textstyle\frac12B^{*2}(T_1,T_2) \phi_r^2(t,T_1)\Sigma_{rr}(0,t), \label{eq:theta}
	\end{align}
	and $\tilde{\mG}_1(t,t_1,T_2)$ is defined analogously to $\mG_1(t,t_1,T_2)$ above except in that
	we re-define
	\begin{align}
		\Delta y(t,t_1)&= \frac{\phi_r(T_1\wedge t_1,T_1\vee t_1) \Sigma_{rr}(t,T_1\wedge t_1)} {\phi_r(t,T_1)},\\
		\Delta z(t,t_1)&= \left(\Sigma_{rz}(T_1,t_1) +B^+(T_1,t_1,T_2)\Sigma_{rr}(T_1,t_1)\right) \mathbbm{1}_{t_1 > T_1}.
	\end{align}
\end{theorem}

\begin{proof}
	The proof of this result is set out in Appendix \ref{sec:CapletPricingProof}. Cap prices are of course obtained as an algebraic sum of caplet prices.
\end{proof}

\paragraph{LIBOR Caplet Pricing}\label{LiborCapletPricing}

LIBOR or term-rate cap and caplet prices are obtained similarly.
\begin{theorem}[LIBOR Caplet Price]\label{LIBORCapletPrice}
Consider a caplet based on the term rate over a payment period $[T_1,T_2]$ and a payoff with strike $K$ at time $T_2$ of
\begin{align}\label{eq:P_LIBOR_caplet}
	P_{\text{LIBOR caplet}}(y_1)&:=\left[1/F^{T_2}(y_1,T_1)-\kappa^{-1}\right]^+
\end{align}
where $y_{T_1}=y_1$. This gives rise to a $T_1$-value of
$$\label{eq:V(y,T_1)}
	V_{\text{LIBOR caplet}}(y,T_1)=\Big[1-\kappa^{-1}F^{T_2}(y,T_1)\Big]^+.$$
Using the above-defined notation, the caplet PV will be given asymptotically with relative errors $=\Oo(\epsilon^2)$ by
\begin{align}
	PV_{\text{LIBOR Caplet}}&\sim D(0,T_1)\left(1-\int_0^{T_1}\tilde{\mG}_1(0,t_1,T_1)  dt_1\right)\Phi(\overline{d}_1(y,0))\Bigg|_{y=0}\cr
	&\quad-\kappa^{-1}D(0,T_2)\left(1-\int_0^{T_2}\tilde{\mG}_1(0,t_1,T_2)e^{-\theta(y,0) H(T_1-t_1)}dt_1 -B^*(T_1,T_2)\Sigma_{rz}(0,T_1)\right)\cr
	&\hspace{320pt}\Phi(\overline{d}_2(y,0))\Bigg|_{y=0}\qquad\cr
	&\quad-\kappa^{-1}D(0,T_2)\sqrt{B^{*2}(T_1,T_2)\Sigma_{rr}(0,T_1)} N(\overline{d}_2(0,0)), \label{eq:PV_LIBOR_caplet}
\end{align}
where we define  additionally
\begin{align}
	\overline{d}_1(y,t)&:=\frac{\theta(y,t) -\Delta z^* +\frac12B^{*2}(T_1,T_2) \Sigma_{rr}(t,T_1)} {\sqrt{B^{*2}(T_1,T_2)\Sigma_{rr}(t,T_1)}}\label{eq:bar_d_1}\\
	\overline{d}_2(y,t)&:=\overline{d}_1(y,t) -\sqrt{B^{*2}(T_1,T_2)\Sigma_{rr}(t,T_1)}. \label{eq:bar_d_2}
\end{align}
\end{theorem}
\begin{proof}
This result can be understood as a special case of Theorem \ref{CapletPrice} when a zero value is assigned to the volatility $\sigma(t)$ for $t\in[T_1, T_2]$.
\end{proof}

\paragraph{Swaption Pricing}\label{SwaptionPricing}

Swaption prices are obtained similarly.
\begin{theorem}[Swaption Price]\label{SwaptionPrice}
	Consider a payer swaption based on a LIBOR or compounded-rate swap with payment periods $[T_{i-1},T_i]$ for $i=1, 2, \ldots,n$ and a fixed coupon $K$ with payoff at time $T_0$ given by the swap value at time $T_0$ if this is in the money.

	Using the above-defined notation, the swaption PV will be given asymptotically with relative errors $=\Oo(\epsilon^2)$ by
	\begin{align}
		PV_{\text{Swaption}}&\sim D(0,T_0)\left(1-\int_0^{T_0}\tilde{\mG}_1(0,t_1,T_0)  dt_1\right)\Phi(d^{(0)}(y))\Bigg|_{y=0}\cr
		&\quad-D(0,T_n)\left(1-\int_0^{T_n}\tilde{\mG}_1(0,t_1,T_n)e^{-B^*(T_0,T_n)\phi_r(0,T_0)y H(T_0-t_1)}dt_1 -B^*(T_0,T_n)\Sigma_{rz}(0,T_0)\right)\cr &\hspace{360pt}\Phi(d^{(n)}(y))\Bigg|_{y=0}\qquad\cr
		&\hspace{75pt}+\sqrt{B^{*2}(T_0,T_n)\Sigma_{rr}(0,T_0)} N(d^{(n)}(0)) \cr
		&\quad-K\sum_{i=1}^n\delta(T_{i-1},T_i)D(0,T_i)\cr
		&\hspace{55pt}\Bigg(\left(1-\int_0^{T_i}\tilde{\mG}_1(0,t_1,T_i)e^{-B^*(T_0,T_i) \phi_r(0,T_0)y H(T_0-t_1)}dt_1 -B^*(T_0,T_i)\Sigma_{rz}(0,T_0)\right)\cr &\hspace{360pt}\Phi(d^{(i)}(y))\Bigg|_{y=0}\qquad\cr
		&\hspace{75pt}+\sqrt{B^{*2}(T_0,T_i)\Sigma_{rr}(0,T_0)} N(d^{(i)}(0))\Bigg), \label{eq:PV_Swaption}
	\end{align}
	where we define additionally for $i=0, 1, \ldots, n$
	\begin{align}
		d^{(i)}(y)&:=\frac{\phi_r(0,T_0)y -y_c-\Sigma_{rz}(0,T_0)} {\sqrt{\Sigma_{rr}(0,T_0)}} -\sqrt{B^{*2}(T_0,T_i)\Sigma_{rr}(0,T_0)},\label{eq:bar_d^(i)}
	\end{align}
	where $y_c$ is the value of $y_{T_0}$ at which the swap comes into the money.
\end{theorem}
\begin{proof}
	The proof of this result is analogous to that provided in Appendix \ref{sec:CapletPricingProof}. This follows from the similarity of the swaption payoff:
	\begin{equation*}
		P_{Swaption}(y,T_0) = 1-   F^{T_{n}}(y,T_0) - K \sum_{i=1}^n \delta(T_{i-1}, T_i)F^{T_{i}}(y,T_0)
	\end{equation*}
	with Eq. \eqref{eq:P_LIBOR_caplet}.
	A key observation is that the value of a risk-free interest rate payment stream, starting at $T_0$ with final payment at $T_n$ and discounted at the risk-free interest rate is $PV = F^{T_0}(y,t)-F^{T_n}(y,t)$, i.e. the difference between the value of a unit payment at time $T_0$ and one at time $T_n$. We note further that this observation is equally true of compounded-rate and term-rate swaptions, since the term-rate payment is by definition the expected value of a compounded-rate payment, observed on its initial fixing date.
\end{proof}

\subsection{Implied Volatility Formulae} \label{subsec:ImpliedVolatility}

We look to translate the formulae set out in \eqref{eq:PV_caplet} and \eqref{eq:PV_LIBOR_caplet} above into expressions for effective/implied Hull-White term variances. This has the advantage of making it more intuitively obvious how the smile and skew impact on the caplet pricing as a function of moneyness. Indeed, one of the first things practitioners typically want to do with option prices is to translate them into effective/implied volatility equivalents. A second advantage is that if, rather than using our asymptotic formulae directly, we use the Hull-White formulae with our best estimate of the effective variance substituted for the Hull-White variance, we naturally obtain better extrapolated behaviour far from the money, with arbitrage avoided.%
\footnote{The quadratic expressions \eqref{eq:v} and \eqref{eq:v_bar} obtained below may arguably give rise for values of $\epsilon$ of order unity to effective variance specifications which are negative. But such circumstances are of no practical relevance and would in any event compromise the asymptotic approach entirely, not just the implied volatility specification.}

We note in passing that one of the main reasons for the popularity of the SABR model of \cite{SABR} was precisely that the asymptotic representation of option prices was explicitly as a formula for the implied Black-Scholes volatility. The approach we take is along the lines of that made use of by \cite{AlosHeston2}.

\paragraph{Effective Compounded Rates Variance}
Starting with the expression \eqref{eq:PV_caplet} for the compounded caplet price, we look to re-express this as
\begin{equation}\label{eq:PV_caplet2}
	PV_{\text{Caplet}}\sim D(0,T_1)\Phi(d_{(1)}(K,T_1,T_2)) -\kappa^{-1}D(0,T_2) \Phi(d_{(2)}(K,T_1,T_2))
\end{equation}
with relative errors $=\Oo(\epsilon^2)$, where
\begin{align}
	d_{(1)}(K,T_1,T_2)) &:=\frac{-\Delta z^* +\frac12\left(V_C +v(K, T_1, T_2)\right)} {\sqrt{V_C +v(K, T_1, T_2)}}\label{eq:d_(1)}\\
	d_{(2)}(K,T_1,T_2)) &:=d_{(1)}(K, T_1, T_2) -\sqrt{V_C + v(K, T_1, T_2)},\label{eq:d_(2)}
\end{align}
for some variance adjustment function $v(K, T_1, T_2)$ defined against a baseline compounded-rate term variance
\begin{equation}\label{eq:V_C}
	V_C=B^{*2}(T_1,T_2)\Sigma_{rr}(0,T_1) +\Sigma_{zz}(T_1,T_2),
\end{equation}
so that the effective variance is
\begin{equation} \label{eq:V_eff}
	V_{\text{effective}} = V_C + v(K, T_1, T_2).
\end{equation}

To determine the appropriate functional form for $v(K, T_1, T_2)$, we expand \eqref{eq:PV_caplet} and \eqref{eq:PV_caplet2} as Taylor expansions and equate terms. To that end we note expanding \eqref{eq:PV_caplet2} to first order in terms of the small parameter $v(K, T_1, T_2)$ that
\begin{align}
	PV_{\text{Caplet}}&\sim D(0,T_1)\Phi(d_{(1)}) -\kappa^{-1}D(0,T_2)\left(\Phi(d_{(2)})\ -\frac{v(K, T_1, T_2) }{2\sqrt{V_C}}N(d_{(2)})\right),
\end{align}
with $d_{(j)}$ given by $d_j(0,0,0)$, while the expansion of \eqref{eq:PV_caplet} gives rise with $\Oo(\epsilon^2)$ relative error to
\begin{align*}
	&D(0,T_1)\Phi(d_{(1)}) -\kappa^{-1}D(0,T_2)\Phi(d_{(2)}) \cr&\quad+\left(C_1(T_1,T_2) -C_2(T_1,T_2)\, d_{(2)} +C_3(T_1,T_2) \left({d_{(2)}}{}^2 -1\right)\right) \kappa^{-1}D(0,T_2) N(d_{(2)})
\end{align*}
where
\begin{align}
	C_1(T_1,T_2)&:=\int_{T_1}^{T_2}\left(\psi_r(0,t_1)\cosh Y^*(t_1,T_2) -\psi_r(T_1,t_1)\right) \Psi_C(T_1,t_1,T_2)\,dt_1\\
	C_2(T_1,T_2)&:=\frac{1}{2}\int_{T_1}^{T_2} \psi_r(0,t_1) \sinh Y^*(t_1,T_2) \gamma(t_1)\Psi_C^2(T_1,t_1,T_2) \,dt_1\\
	C_3(T_1,T_2)&:=\frac{1}{3!}\int_{T_1}^{T_2} \psi_r(0,t_1) \cosh Y^*(t_1,T_2)\gamma^2(t_1)\Psi_C^3(T_1,t_1,T_2)\,dt_1,\label{eq:C_{3,2}}
\end{align}
with
\begin{align}
	\Psi_C(T_1,t_1,T_2)&:=V_C^{-\frac12}\left(B^*(T_1,T_2)\phi_r(T_1,t_1) \Sigma_{rr}(0,T_1) +\Sigma_{rz}(T_1,t_1) +B^+(T_1,t_1,T_2)\Sigma_{rr}(T_1,t_1) \right)
\end{align}
and $Y^*(\cdot,\cdot)$ defined by \eqref{eq:Y^*} above.

Equating terms in the two expansions, we conclude that we must have
\begin{align}\label{eq:v}
	v(K, T_1, T_2) &\sim 2V_C^{\frac12}\left(C_1(T_1,T_2) -C_2(T_1,T_2)\, d_{(2)} +C_3(T_1,T_2) \left({d_{(2)}}{}^2 -1\right)\right),
\end{align}
from which we deduce an effective compounding rate caplet term variance given by \eqref{eq:V_eff} with $\Oo(\epsilon^2)$ relative error. It is at this point evident how the linear term involving $C_2(T_1,T_2)$ gives rise to the skew and the quadratic term involving $C_3(T_1,T_2)$ to the smile. We suggest on this basis that taking the parameter $\epsilon$ to be given here by $\max\{|C_2(T_1,T_2)|,C_3(T_1,T_2)\}/V_C^{\frac12}$ would furnish a good guide as to the level of accuracy of the first order effective variance expansion, with relative errors being expected to be of magnitude $\epsilon^2$.

 \paragraph{Effective LIBOR Variance}
Starting instead with the expression \eqref{eq:PV_LIBOR_caplet} for the LIBOR or term-rate caplet price, we look to re-express this as
\begin{equation}\label{eq:PV_LIBOR_caplet2}
	PV_{\text{LIBOR Caplet}}\sim D(0,T_1)\Phi(\overline{d}_{(1)}(K,T_1,T_2)) -\kappa^{-1}D(0,T_2) \Phi(\overline{d}_{(2)}(K,T_1,T_2))
\end{equation}
with relative errors $=\Oo(\epsilon^2)$, where
\begin{align}
	\overline{d}_{(1)}(K,T_1,T_2)) &:=\frac{-\Delta z^* +\frac12\left(V_L + \overline{v}(K, T_1, T_2)\right)} {\sqrt{V_L + \overline{v}(K, T_1, T_2)}}\label{eq:bar_d_(1)}\\
	\overline{d}_{(2)}(K,T_1,T_2)) &:=\overline{d}_{(1)}(K, T_1, T_2) -\sqrt{V_L + \overline{v}(K, T_1, T_2)},\label{eq:bar_d_(2)}
\end{align}
for some variance adjustment function $\overline{v}(K, T_1, T_2)$ defined against a baseline LIBOR term variance
\begin{equation}\label{eq:V_L}
	V_L:=B^{*2}(T_1,T_2)\Sigma_{rr}(0,T_1),
\end{equation}
so that the effective variance is
\begin{equation}\label{eq:V_eff2}
	V_{\text{effective}} = V_L + \overline{v}(K, T_1, T_2)
\end{equation}
By a calculation analogous to the above, we conclude that we must have
\begin{align}\label{eq:v_bar}
	\overline{v}(K, T_1, T_2) &\sim 2V_L^{\frac12}\left(\overline{C}_1(T_1,T_2) -\overline{C}_2(T_1,T_2)\,  \overline{d}_{(2)} +\overline{C}_3(T_1,T_2) \left({\overline{d}_{(2)}}^2 -1\right)\right),
\end{align}
with $\overline{d}_{(j)}$ given by $\overline{d}_j(0,0)$ and
\begin{align}
	\overline{C}_1(T_1,T_2)&:=\int_{T_1}^{T_2} \left(\psi_r(0,t_1)\cosh Y^*(t_1,T_2) -\psi_r(T_1,t_1)\right)\Psi_L(T_1,t_1)\,dt_1,\\
	\overline{C}_2(T_1,T_2)&:=\frac{1}{2}\int_{T_1}^{T_2} \psi_r(0,t_1) \sinh Y^*(t_1,T_2)\gamma(t_1)\Psi_L^2(T_1,t_1)\,dt_1,\\
	\overline{C}_3(T_1,T_2)&:=\frac{1}{3!}\int_{T_1}^{T_2} \psi_r(0,t_1)\cosh Y^*(t_1,T_2)\gamma^2(t_1)\Psi_L^3(T_1,t_1)\,dt_1,
\end{align}
with
\begin{align}
	\Psi_L(T_1,t_1)&:=\phi_r(T_1,t_1) \sqrt{\Sigma_{rr}(0,T_1)}.
\end{align}

\paragraph{Effective Swaption Variance}

Starting instead with the expression \eqref{eq:PV_Swaption} for the swaption price, we look to re-express this as
\begin{align}	\label{eq:PV_Swaption2}
	PV_{\text{Swaption}}&\sim D(0,T_0)\Phi(\tilde{d}^{(0)}(K, \bold{T})) -D(0,T_n)\Phi(\tilde{d}^{(n)}(K, \bold{T}))-K\sum_{i=1}^n\delta(T_{i-1},T_i)D(0,T_i)\Phi(\tilde{d}^{(i)}(K, \bold{T})),
\end{align}
with relative errors $=\Oo(\epsilon^2)$, where
\begin{align}
	\tilde{d}^{(i)}(K,\bold{T})) &:=\frac{-y_c-\Sigma_{rz}(0,T_0)} {\sqrt{\Sigma_{rr}(0,T_0) + \tilde{v}(K, \bold{T})}}-B^*(T_0,T_i)\sqrt{\Sigma_{rr}(0,T_0)+\tilde{v}(K, \bold{T})}\label{eq:tilde_d^(i)}
\end{align}
for some variance adjustment function $\tilde{v}(K, \bold{T})$ with $\bold{T} := (T_0,T_1, \ldots, T_n)$, with the effective short rate term variance taken to be
\begin{equation}\label{eq:V_eff3}
	V_{\text{effective}} = \Sigma_{rr}(0,T_0) +\tilde{v}(K, \bold{T}).
\end{equation}
Defining $\tilde{d}_{(i)}$ to be the result obtained by setting $\tilde{v}(K, \bold{T})=0$ in \eqref{eq:tilde_d^(i)}, we can expand \eqref{eq:PV_Swaption2} as
\begin{align} \label{eq:PV_Swaption3}
	PV_{\text{Swaption}}&\sim D(0,T_0)\Phi\left(\tilde{d}_{(0)}\right) -D(0,T_n)\Phi\left(\tilde{d}_{(n)}\right)-K\sum_{i=1}^n\delta(T_{i-1},T_i)D(0,T_i)\Phi\left(\tilde{d}_{(i)}\right)\cr
	&\quad+\frac{\tilde{v}(K, \bold{T})}{2\sqrt{\Sigma_{rr}(0,T_0)}} \sum_{i=1}^n \left(\delta_{in} +K\delta(T_{i-1},T_i)\right) \left(1 +\Delta\tilde{d}_{(i)} \tilde{d}_{(n)} \right) B^*(T_0,T_i) D(0,T_i) N\left(\tilde{d}_{(n)}\right),\qquad
\end{align}
where we have made use of the fact that
\begin{equation}\label{eq:N}
	N\left(\tilde{d}_{(i)}\right) \sim\left(1 +\Delta\tilde{d}_{(i)} \tilde{d}_{(n)} \right) N\left(\tilde{d}_{(n)}\right)
\end{equation}
with
\begin{equation}\label{eq:Delta d_(i)}
	\Delta\tilde{d}_{(i)}:=\tilde{d}_{(n)}-\tilde{d}_{(i)}
\end{equation}
and $\delta_{in}$ the Kronecker delta. We observe that the sum in the second line of \eqref{eq:PV_Swaption3} can be written more compactly as
$$\left(A_n+B_n \tilde{d}_{(n)}\right) N\left(\tilde{d}_{(n)}\right)$$
Expanding \eqref{eq:PV_Swaption} as in the previous cases leads to the conclusion that the second line in \eqref{eq:PV_Swaption3} must be equated with
$$\sum_{i=1}^n \left(\overline{C}_1(T_0,T_i) -\overline{C}_2(T_0,T_i)\,  \tilde{d}_{(i)} +\overline{C}_3(T_0,T_i) \left(\tilde{d}_{(i)}^2 -1\right)\right) \left(\delta_{in} +K\delta(T_{i-1},T_i)\right) D(0,T_i) N\left(\tilde{d}_{(i)}\right).$$
This can be written asymptotically, making use of \eqref{eq:N}, as
\begin{align*}
	\sum_{i=1}^n &\left(\overline{C}_1(T_0,T_i) -\overline{C}_2(T_0,T_i)\, \left(\tilde{d}_{(n)}-\Delta\tilde{d}_{(i)}\right) +\overline{C}_3(T_0,T_i) \left(\left(\tilde{d}_{(n)}-\Delta\tilde{d}_{(i)}\right)^2 -1\right)\right)\cr
	&\hspace{120pt}\left(\delta_{in} +K\delta(T_{i-1},T_i)\right)\left(1 +\Delta\tilde{d}_{(i)} \tilde{d}_{(n)} \right) D(0,T_i)  N\left(\tilde{d}_{(n)}\right),
\end{align*}
which we write more compactly as
$$\left(D_n +E_n\tilde{d}_{(n)} +F_n\left(\tilde{d}_{(n)}^2- 1\right)\right) N\left(\tilde{d}_{(n)}\right).$$
Asymptotic matching then requires that $\tilde{v}(K, \bold{T})$ be chosen to satisfy
\begin{equation}\label{eq:tilde_v}
	\tilde{v}(K, \bold{T})\sim\frac{2\left(D_n +E_n\tilde{d}_{(n)} +F_n\left(\tilde{d}_{(n)}^2- 1\right)\right)}{\left(A_n+B_n \tilde{d}_{(n)}\right) \sqrt{\Sigma_{rr}(0,T_0)}}.
\end{equation}
As can be seen, the effective term variance is on this occasion not quite quadratic in its dependence on log-moneyness. Further, the form of the denominator will cause it to give rise to singular behaviour for values of $\tilde{d}_{(n)}$ sufficiently far from the ATM level. However, since the model would not in any event be calibrated to give credible results in such extreme cases, this ought not to be a significant limitation in practice.

\section{Numerical Calculations} \label{sec:Numerical}
\paragraph{Caplets}

The above model has been calibrated to caplet market data capturing the skew and smile for SAR 6M LIBOR rates in May 2021. For the mean reversion, we choose a representative fixed value of $\alpha(t) = 0.15$. Choosing other values made little difference to the results obtained. Since the model has three other disposable parameters ($\sigma(t)$, $y^*(t)$ and $\gamma(t)$) it should be possible making use of \eqref{eq:v} to match the ATM level, smile and skew of the implied volatility surface at each maturity for which market data are provided. Indeed this is found to be the case.%
\footnote{In fact the fitting was done to LIBOR caplet data treated as compounded-rate caplets, since the latter were not available and this approach satisfied the exigency of convenient presentation of numerical results for compounded-rate formulae. In any event, it was found to be no more difficult to calibrate the model interpreting the data as being for LIBOR caplets.}
The resulting fit is illustrated in Fig.~\ref{fig:implVol} and is seen to be excellent. Results are expressed as $\sigma_{IV}(K,T_2)$, the constant value of $\sigma(t)$ which would have to be inserted in the Hull-White formula to replicate the price of a (6M tenor) caplet with strike $K$ and maturity $T_2$. For the smile, calibrated values of $\gamma(t)$ ranged from around 300 at the short end to 40 for $t = 10$y while, for the skew, the corresponding values of $\gamma(t)y^*(t)$ were 0.5 and 0.08. The corresponding implied volatility surface is illustrated in  Fig.~\ref{fig:impvolsurface}.

\begin{figure}[H]
	\centering
	\includegraphics[width=0.75\linewidth]{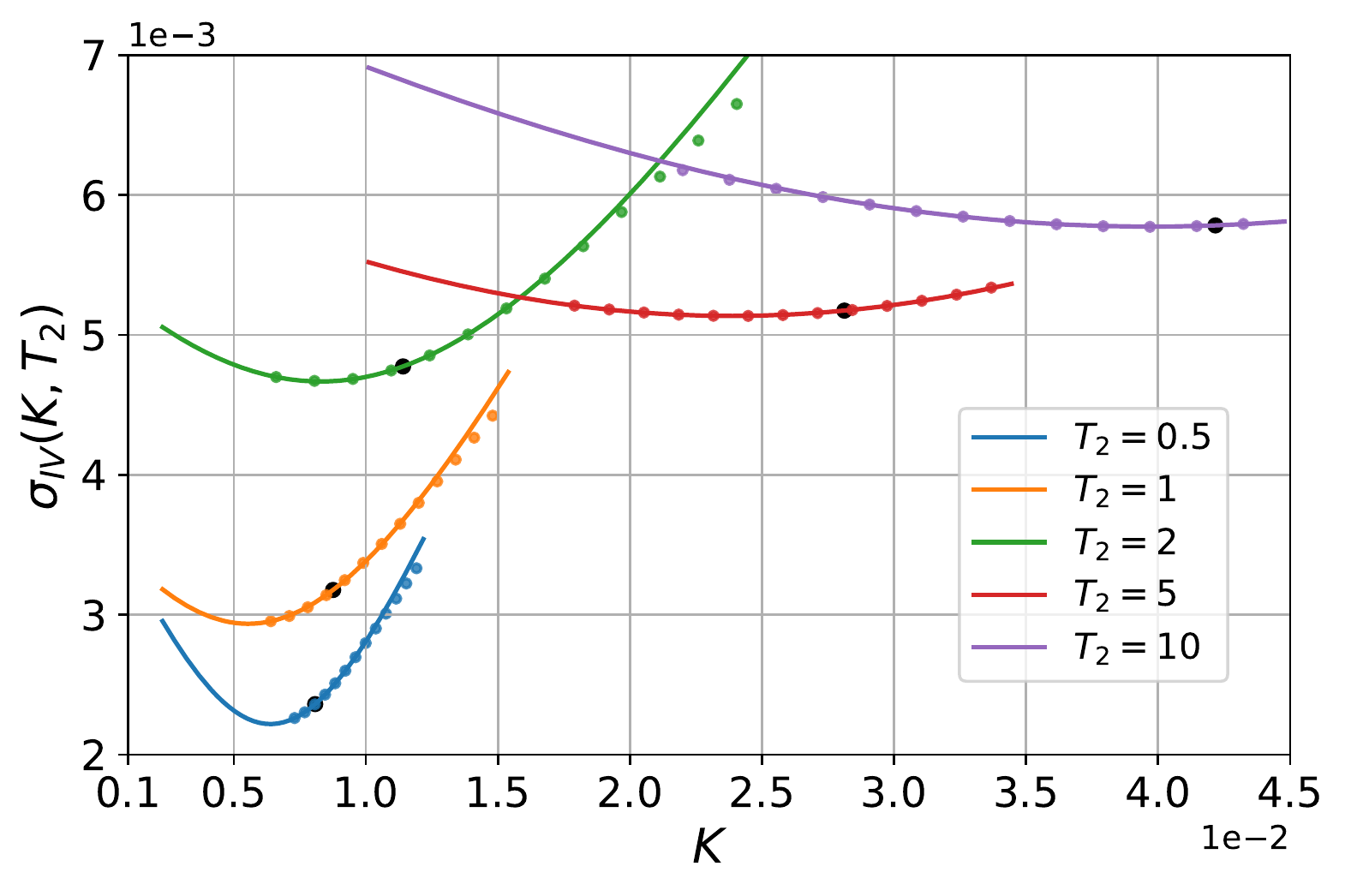}
	\caption{Implied volatilities for various caplet maturities $T_2$}
	\label{fig:implVol}
\end{figure}

\begin{figure}[H]
	\centering
	\includegraphics[width=0.75\linewidth]{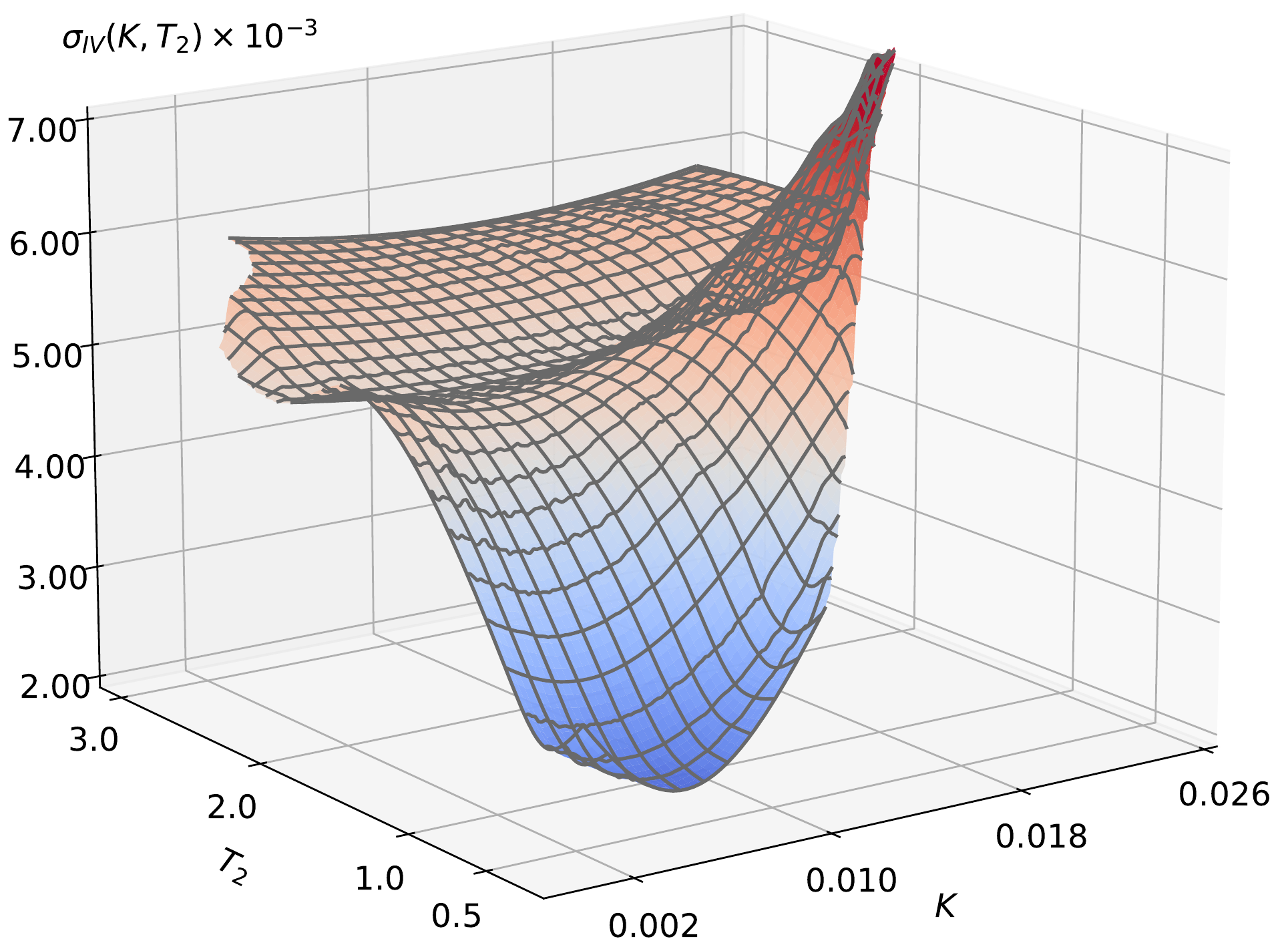}
	\caption{Caplet implied volatility surface}
	\label{fig:impvolsurface}
\end{figure}

Values of $\sigma_{IV}(K,T_2)$ were also explicitly backed out from the compounded-rate caplet price formula \eqref{eq:PV_caplet}. These are shown for comparison in Fig.~\ref{fig:implVol2}. As can be seen, the fit is similarly good close to the money but the asymptotic caplet price formula clearly becomes unreliable for more extreme strikes so caution should be exercised in its direct use to calculate prices. Rather, the use of an asymptotic representation of the \emph{implied volatility surface}, embedding as it does the avoidance of arbitrage (negative option prices) for any reasonable strikes, is to be preferred.%
\footnote{For the same reason analytic option prices under the SABR model are always constructed indirectly from asymptotic representation of implied volatilities.}

\begin{figure}[H]
	\centering
	\includegraphics[width=0.75\linewidth]{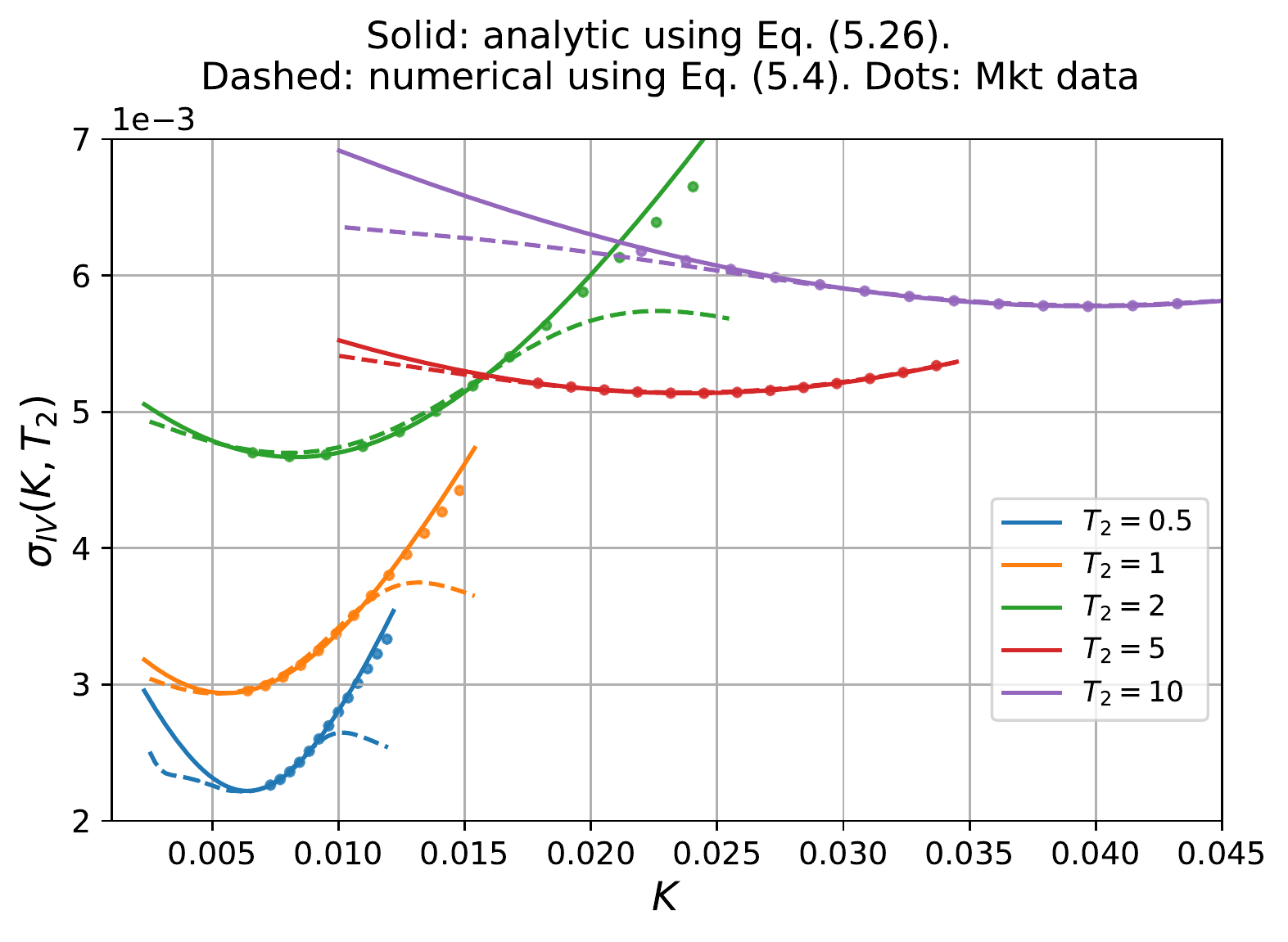}
	\caption{Implied volatilities for various maturities}
	\label{fig:implVol2}
\end{figure}

It is of interest to consider how much impact results from use of a compounded rate rather than a term rate as the caplet underlying. This is illustrated in Figs.~\ref{fig:pvcapletcomparecmplibort212} and \ref{fig:pvcapletcomparecmplibort2510} where PVs based on formulae \eqref{eq:PV_caplet2} and \eqref{eq:PV_LIBOR_caplet2} are compared. As is evident, the impact is much greater for shorter maturities where $T_2-T_1$ is comparable with $T_1$, and becomes rather insignificant when $T_2-T_1\ll T_1$

\begin{figure}[H]
	\centering
	\includegraphics[width=0.7\linewidth]{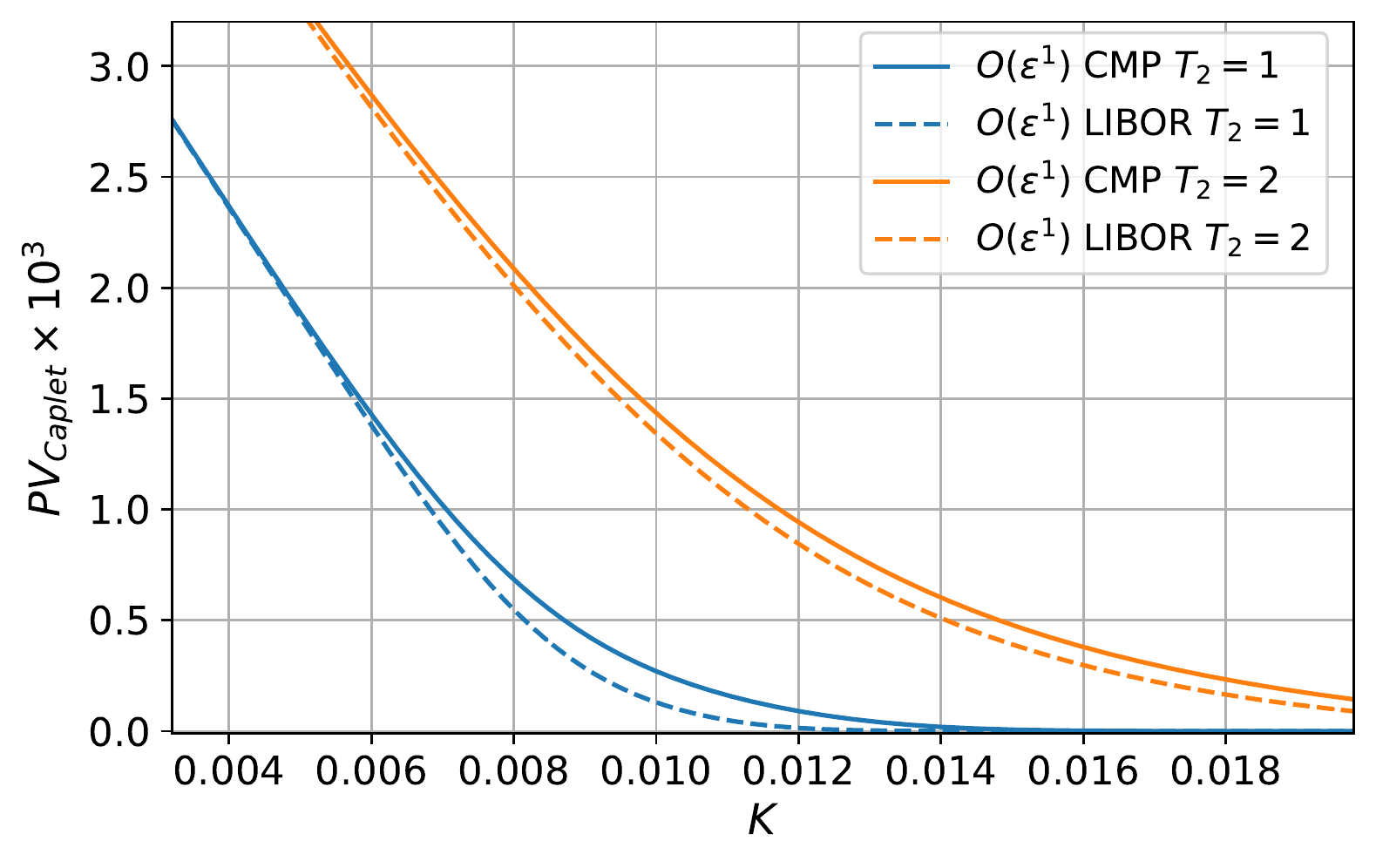}
	\caption[fig6]{Comparison of term-rate and compounded-rate caplet prices for short maturities}
	\label{fig:pvcapletcomparecmplibort212}
\end{figure}

\begin{figure}[H]
	\centering
	\includegraphics[width=0.7\linewidth]{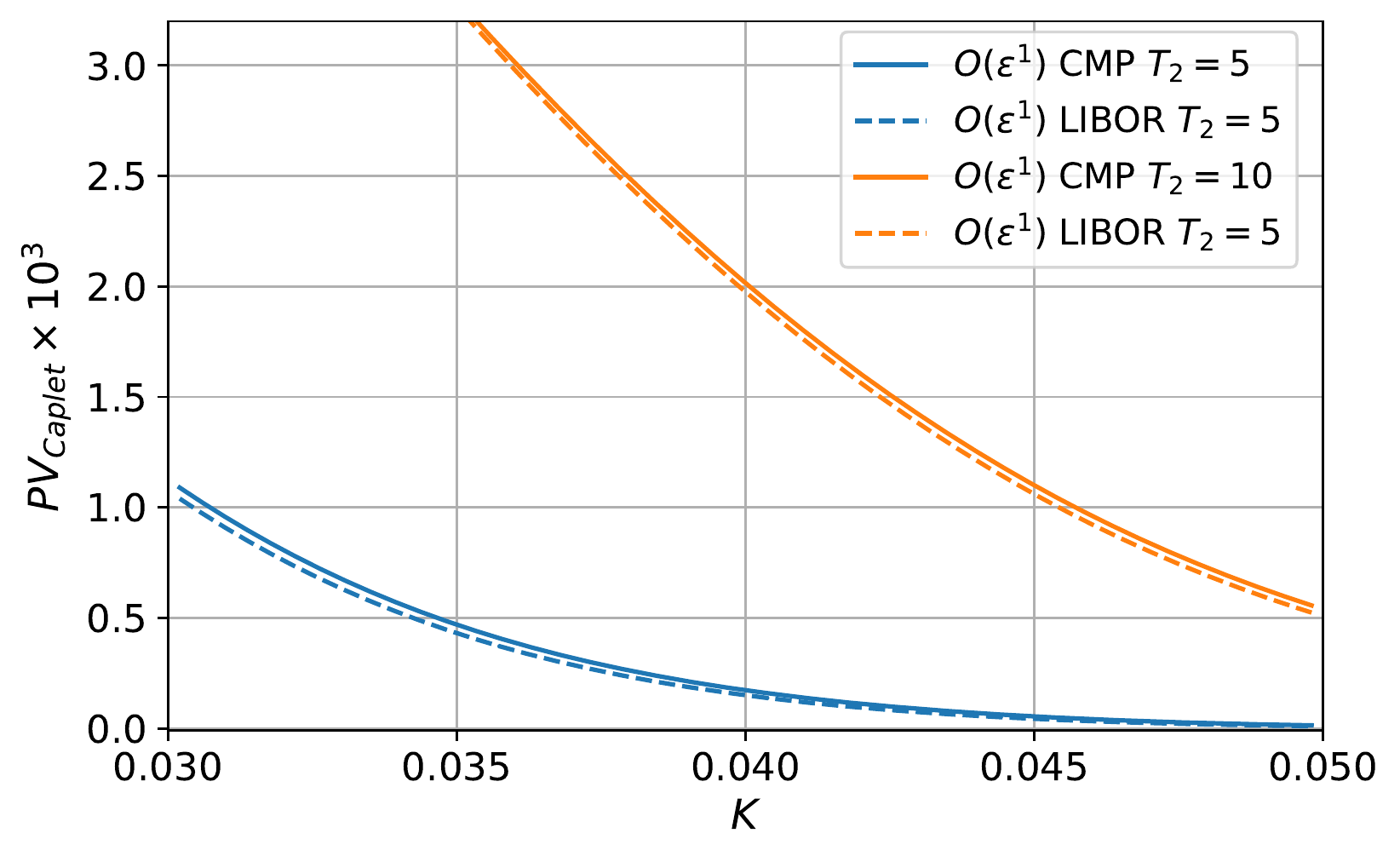}
	\caption{Comparison of term-rate and compounded-rate caplet prices for longer maturities}
	\label{fig:pvcapletcomparecmplibort2510}
\end{figure}

For completeness we also looked at instantaneous forward rates calculated in accordance with \eqref{eq:f^T}, with fixing date $T = t + 0.5$. Results are shown in Fig.~\ref{fig:fwdRates}. Since these results incorporate \emph{second} order terms, unlike the caplet results shown above, they are expected to be highly accurate. The impact of the skew/smile can be inferred by comparing with the dashed lines (Hull-White). As expected, the model is seen to behave just like Hull-White with forward rates linear in the underlying Gaussian variable $y$ when this is small in magnitude.
\begin{figure}[H]
	\centering
	\includegraphics[width=0.75\linewidth]{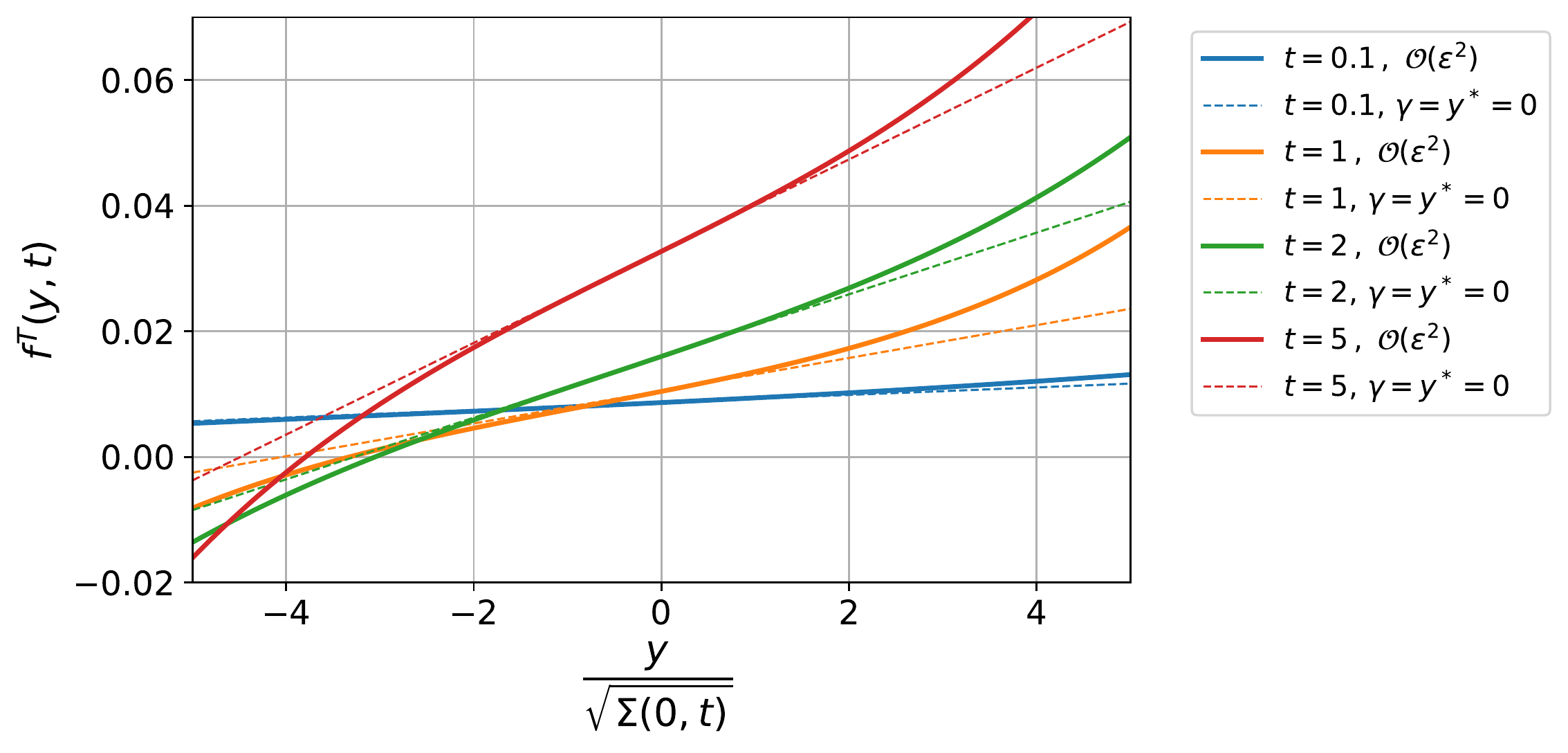}
	\caption{Forward rates for various observation dates $t$ with payment date $T=t+0.5$.}
	\label{fig:fwdRates}
\end{figure}

\paragraph{Swaptions}
The implied volatility surface of swaptions with two payment periods of 3 months is shown in Fig. \ref{fig:swaption}.
\begin{figure}[h]
\subfigure{\includegraphics[width=0.49\linewidth]{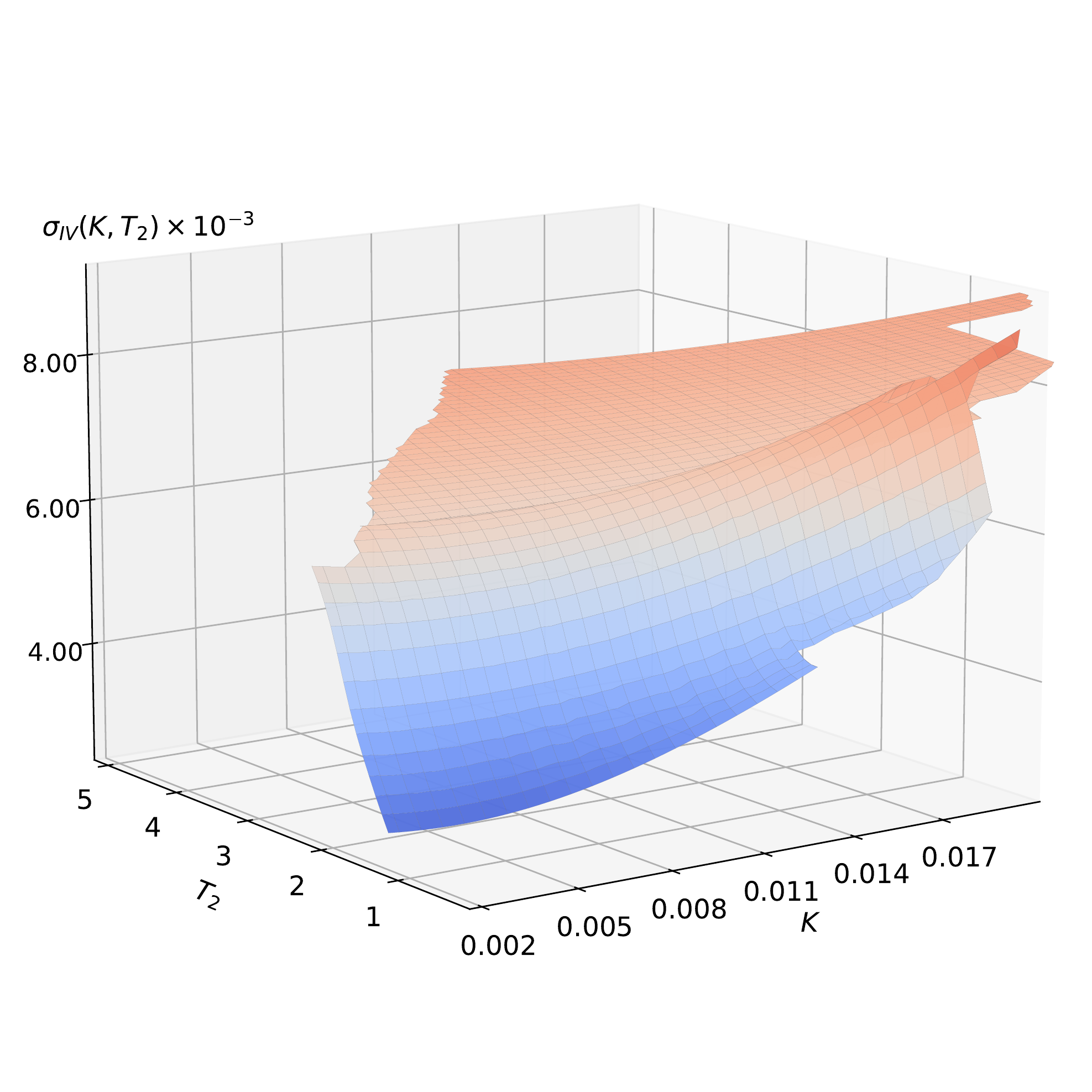}}
\subfigure{\raisebox{10mm}{\includegraphics[width=0.49\linewidth]{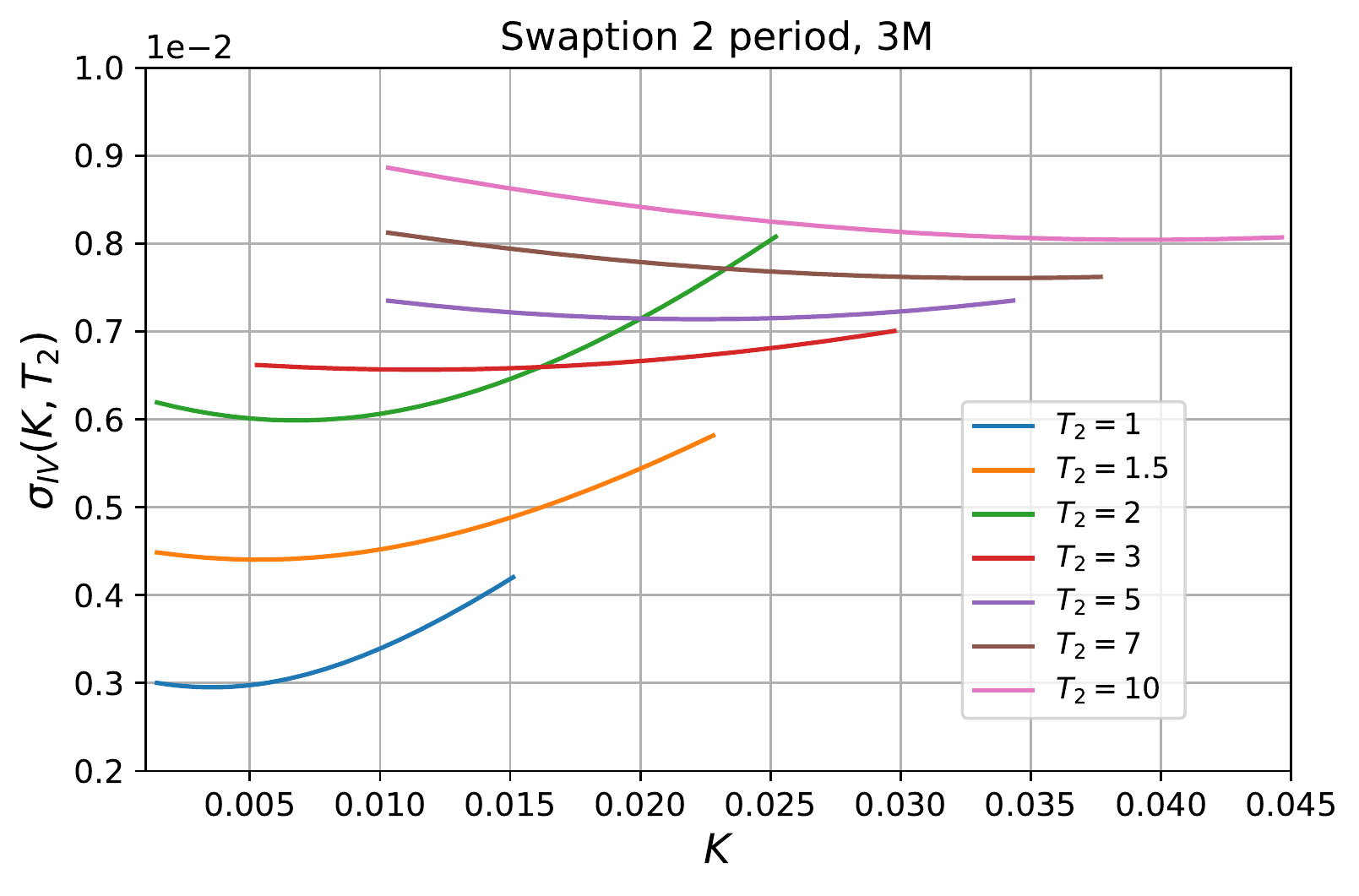}}}
\vspace{-1cm}
\caption{Swaption implied volatilities for various maturities}\label{fig:swaption}
\end{figure}

\newpage

\section{Conclusions}\label{sec:Conclusions}

We have successfully extended the short rate model of \cite{SmilesWithoutTears} to address the pricing of SOFR/SONIA/ESTR caplets based on compounded rates. This we achieved by expressing the result as a perturbation of the analytic kernel of \cite{Compounded_Rates}. The model is seen to have the following properties:
\begin{enumerate}
	\item Convenient analytic representations are available for bond and option prices.
	\item A single calibration addresses options of any maturity and tenor.
	\item Option pricing takes account of volatility skew and smile.
	\item Prices can be calculated for LIBOR or term-rate options as well as for options on backward-looking compounded rates.
\end{enumerate}
We believe our model to be unique in satisfying all of the above criteria.

The asymptotic expansion developed here provides the further benefit of pricing forward rates and LIBOR caplets more accurately than that of \cite{SmilesWithoutTears}. The expression for forward rates is \eqref{eq:f^T}. Caplet prices are best calculated noting that the effective Hull-White term variance is given by \eqref{eq:V_eff} and \eqref{eq:v} for a compounded-rate caplet and by \eqref{eq:V_eff2} and \eqref{eq:v_bar} for a LIBOR or term-rate caplet.

\appendix

\section{Proof of Pricing Kernel Result }\label{sec:PricingKernelProof}
We offer a sketch of the proof that \eqref{eq:G_compounded} represents the pricing kernel for \eqref{eq:fPDE}. We follow \cite{Exact_H-W,PerturbationMethods} in introducing the following definitions and notation.

\begin{definition}\label{Dyson}
	The \emph{time-ordered exponential} or \emph{Dyson series} defined for a linear operator\index{operator!linear} $\mL(t)$ by
	\begin{align}\label{eq:Dyson}
		\mE_t^T (\mL(\cdot))&=I+\sum_{n=1}^\infty \int_{t\le t_1\le \ldots\le t_n\le T}
		\mL(t_1)\ldots\mL(t_n)dt_1\ldots dt_n\cr
		&=I+\sum_{n=1}^\infty \int_t^T\int_{t_1}^T\cdots\int_{t_{n-1}}^T
		\mL(u)\mL(t_2)\ldots\mL(t_{n})dt_n\ldots dt_2dt_1\cr
		&=I+\sum_{n=1}^\infty \int_t^T\int_t^{t_n}\cdots\int_t^{t_2}
		\mL(t_1)\mL(t_2)\ldots\mL(t_{n})dt_1dt_2\ldots dt_n
	\end{align}
	generalises the exponentiation of the integral of a function $f(t)$ to the case of a time-dependent linear operator\index{operator!linear} $\mL(t)$.
\end{definition}
\begin{definition}\label{def:commutator}
	We define the \emph{commutator}\index{commutator} of an operator $\mL(t)$ with an operator $\mV(t)$ where both operators act on the same function space by the following operator:
	\begin{align}
		\ad_{\mL(t_1)}(\mV(t_2))&:=\mL(t_1)\mV(t_2)-\mV(t_2)\mL(t_1).\label{eq:ad}
	\end{align}
\end{definition}
We will have need to combine these two ideas in computing the time-ordered exponential of a commutator operator, interpreted as follows:
\begin{equation}\label{eq:exp_commutator}
	\mE_t^u\left(\ad_{\mL(\cdot)}\right)(\mV(u))=\mV(u) +\int_t^u \ad_{\mL(t_1)}(\mV(u))dt_1 +\int_t^u\int_{t_2}^u \ad_{\mL(t_2)}\left(\ad_{\mL(t_1)}(\mV(u))\right)dt_1dt_2+\ldots.
\end{equation}

\begin{proof}[Proof of Theorem \ref{PricingKernel}]
We start by observing that, although $r(y,t)-\overline{r}(t)=\Oo(\epsilon^{\frac12})$, the term $(r(y,t)-\overline{r}(t))\partial/\partial z$ is $\Oo(1)$ so prevents direct use of a na\"{i}ve perturbation expansion. Noting however that $r(y,t)-\overline{r}(t)\sim y$ for small $y$ we can subtract out the asymptotic representation and look to handle it separately. To that end we write the evolution operator associated with \eqref{eq:fPDE} as $\mE_t^v(\mL_0(\cdot)+\mV_0(\cdot))$ where
\begin{align}
	\mL_0(t)&=-\alpha(t)y\frac{\partial}{\partial  y} +\frac12\sigma^2(t) \frac{\partial^2 }{\partial y^2} - \overline{r}(t),\\
	\mV_0(t)&= \left(r(y,t)- \overline{r}(t) \right)\left(\frac{\partial}{\partial z}-1\right).
\end{align}
Following \cite{SmilesWithoutTears} who apply the exponential expansion formula of \cite{PerturbationMethods} to a closely analogous evolution operator, we can write the Green's function solution of \eqref{eq:fPDE} as
\begin{align}
	G(y,z,t;\eta,\zeta,v)&=D(t,v)\mE_t^v(\mW(t,\cdot))N\left(\frac{\eta-\phi_r(t,v)y} {\sqrt{\Sigma_{rr}(t,v)}}\right),
\end{align}
where%
\footnote{Observe that the only difference at this stage compared to their result is in the inclusion of the derivative w.r.t.~$z$ in \eqref{eq:W}. However there is an important difference in the solution strategy. By subtracting out of \eqref{eq:V_2} below a Hull-White (linear) representation of the discounting rate in addition to the $z$-drift, we are able to incorporate \emph{both} directly into the analytic kernel. So, rather than the stochastic discounting being represented by a perturbation as in \cite{SmilesWithoutTears}, it is only the difference between that and the Hull-White representation thereof which is so represented.}
\begin{align}
	\mW(t,u)&:=\left(R^+(y,t,u)e^{\gamma(u)\Delta y(t,u)\frac{\partial}{\partial y}} - R^-(y,t,u)e^{-\gamma(u)\Delta y(t,u)\frac{\partial}{\partial y}}+R^*(u)\right) \left(\frac{\partial}{\partial z}-1\right),\label{eq:W}
\end{align}
with
\begin{align}\label{eq:R^+-}
	R^{\pm}(y,t,u)&:=\frac{e^{\frac12\gamma^2(u)\Sigma_{rr}(t,u)\pm\gamma(u)(\phi_r(t,u)y +y^*(u))}} {2\gamma(u)},
\end{align}
Let us define
\begin{align}
	\mL_1(t,u)&:=\psi_r(t,u)\left(\frac{\Sigma_{rr}(t,u)}{\phi_r(t,u)}\frac{\partial}{\partial y} +\Sigma_{rz}(t,u)\frac{\partial}{\partial z} \right)\left(\frac{\partial}{\partial z}-1\right). \label{eq:L_1}\\
	\mV_1(t,u)&:=\mW(t,u)-\mL_1(t,u).
\end{align}
We deduce
\begin{align}
	\mE_t^v(\mW(t,\cdot))&=\mE_t^v(\mW_1(t,\cdot))\mE_t^v(\mL_1(t,\cdot)),\\
	\mW_1(t,u)&=\mE_t^u\left(\ad_{\mL_1(t,\cdot)}\right)(\mV_1(t,u))\cr
	&=\Big(R^+(y,t,u)e^{\gamma(u)\left(\Delta y(t,u)\frac{\partial}{\partial y} +\Sigma_{rz}(t,u)\left(\frac{\partial}{\partial z} -1\right)\right)}\cr
	&\qquad -R^-(y,t,u) e^{-\gamma(u)\left(\Delta y(t,u)\frac{\partial}{\partial y} +\Sigma_{rz}(t,u)\left(\frac{\partial}{\partial z} -1\right)\right)} +R^*(u)\Big)\left(\frac{\partial}{\partial z} -1\right)-\mL_1(t,u).
	\label{eq:W_1}
\end{align}
Next defining
\begin{align}
	\mL_2(t,u)&:=\frac{\partial}{\partial u}\mu^*(y,t,u)  \left(\frac{\partial}{\partial z}-1\right)\cr
	&\,=\psi_r(t,u)\left(\phi_r(t,u)(y +\Sigma_{rz}(0,t) +B^*(t,u) \Sigma_{rr}(0,t))\right) \left(\frac{\partial}{\partial z}-1\right), \qquad\label{eq:L_2}\\
	\mV_2(t,u)&:=\mW_1(t,u)-\mL_2(t,u).\label{eq:V_2}
\end{align}
we deduce
\begin{align}
	\mE_t^v(\mW_1(t,\cdot))&=\mE_t^v(\mW_2(t,\cdot))\mE_t^v(\mL_2(t,\cdot)),\\
	\mW_2(t,u)&=\mE_t^u\left(\ad_{\mL_2(t,\cdot)}\right)(\mV_2(t,u))\cr
	&=\Big(R^+(y,t,u)e^{\gamma(u)\left(\Delta y(t,u)\frac{\partial}{\partial y}+\Delta z(t,u)\left(\frac{\partial}{\partial z} -1\right)\right)}\cr
	&\qquad -R^-(y,t,u) e^{-\gamma(u)\left(\Delta y(t,u)\frac{\partial}{\partial y} +\Delta z(t,u)\left(\frac{\partial}{\partial z} -1\right)\right)} +R^*(u)\Big)\left(\frac{\partial}{\partial z} -1\right) -\mL_3(t,u).\qquad
	\label{eq:W_2}
\end{align}
where
\begin{align}\label{eq:L_3}
	\mL_3(t,u) &:=\mE_t^u\left(\ad_{\mL_2(t,\cdot)}\right)(\mL_1(t,u)+\mL_2(t,u))\cr
	&=\mL_1(t,u) +\mL_2(t,u) -B^*(t,u)\psi_r(t,u) \frac{\Sigma_{rr}(t,u)}{\phi_r(t,u)} \left(\frac{\partial}{\partial z}-1\right)^2.\cr
\end{align}
We observe, making use of the identity%
\footnote{This is readily proved by differentiating both sides w.r.t.~$v$.}
\begin{align}
	\int_t^v B^+(t,u,v)\psi_r(t,u)\Sigma_{rr}(t,u)du&=\frac12\Sigma_{zz}(t,v),
	\label{eq:B^+_identity}
\end{align}
that
\begin{align}
	\int_t^v B^*(t,u) \psi_r(t,u) \frac{\Sigma_{rr}(t,u)}{\phi_r(t,u)}du &= B^*(t,v)\frac{\Sigma_{rz}(t,v)}{\phi_r(t,v)} -\frac12\Sigma_{zz}(t,v),
\end{align}
whence we deduce
\begin{align*}
	\int_t^v \mL_3(t,u)du &= \left(\mu^*(y,t,v) +\frac{\Sigma_{rz}(t,v)}{\phi_r(t,v)} \frac{\partial}{\partial y} +\frac12\Sigma_{zz}(t,v)\frac{\partial}{\partial z}\right.\cr 	&\hspace{30pt}\left.+\left(\frac12\Sigma_{zz}(t,v) -B^*(t,v) \frac{\Sigma_{rz}(t,v)}{\phi_r(t,v)}\right)\left(\frac{\partial}{\partial z}-1\right)\right)
\left(\frac{\partial}{\partial z}-1\right).
\end{align*}
Applying $\mE_t^v(\mL_2(t,\cdot))\mE_t^v(\mL_1(t,\cdot))$ to the Gaussian kernel, so increasing its dimension from 1 to 2, gives rise to
\begin{equation}
	G(y,z,t;\eta,\zeta,v)=D(t,v)\mE_t^v(\mW_2(t,\cdot))e^{-\mu^*(y,t,v)}G_0(y,z,t;\eta,\zeta,v),
\end{equation}
with $G_0(\cdot;\cdot)$ defined by \eqref{eq:G_0} above. It is convenient to move the exponential function to the left of this expression. To this end we note that the exponential of a derivative acts as a shift operator, whence
\begin{align*}
	e^{\pm\gamma(u)\Delta y(t,u)\frac{\partial}{\partial y}}e^{-\mu^*(y,t,v)}&=e^{-\mu^*(y,t,v)}e^{\pm\gamma(u) \left(\frac{\partial}{\partial y} -B^*(t,v) \right)\Delta y(t,u)},\\
	\mL_3(t,u)\, e^{-\mu^*(y,t,v)}&= e^{-\mu^*(y,t,v)}\left(\mL_3(t,u) -B^*(t,v) \frac{\Sigma_{rr}(t,u)}{\phi_r(t,u)}\left(\frac{\partial}{\partial z} -1\right)\right).
\end{align*}
We obtain
\begin{align}
	G(y,z,t;\eta,\zeta,v)&=D(t,v)e^{-\mu^*(y,t,v)}\mE_t^v(\mW_3(t,\cdot,v))G_0(y,z,t;\eta,\zeta,v), \label{eq:G}\\
	\mW_3(t,u,v)&=\Big(R_1^+(y,t,u,v)\mM^+(t,u) -R_1^-(y,t,u,v) \mM^-(t,u) +R^*(u)\Big)\left(\frac{\partial}{\partial z} -1\right)\cr
	&\quad-\mL_3(t,u) +B^*(t,v) \psi_r(t,u) \frac{\Sigma_{rr}(t,u)} {\phi_r(t,u)} \left(\frac{\partial}{\partial z} -1\right),
\end{align}
with $R_1^\pm(y,t,u,v)$ defined by \eqref{eq:R_1^+-}, where we have made use of the fact that
$$e^{\pm\gamma(u)\left(\Delta y(t,u)\frac{\partial}{\partial y} +\Delta z(t,u)\frac{\partial}{\partial z}\right)}\equiv \mM^\pm(t,u).$$
Expanding $\mE_t^v(\mW_3(t,\cdot,v))$ as a power series in \eqref{eq:G} and making use of the above identities, we obtain \eqref{eq:G_compounded}, with \eqref{eq:G_0}--\eqref{eq:G_2} as the first three terms, as claimed. This completes the proof.
\end{proof}

\section{Proof of RFR Caplet Pricing Result}\label{sec:CapletPricingProof}

\begin{proof}[Proof of Theorem \ref{CapletPrice}]
	Making use of \eqref{eq:G_compounded}, the first order caplet value as of time $T_1$ with $y_{T_1}=y$ and $z_{T_1}=z_1$ will be
	\begin{align}
		V_{\text{caplet}}(y,T_1)&=\lim\limits_{z\to z_1}\iint_{\R^2} G(y,z,T_1;\eta,\zeta,T_2) P_{\text{caplet}}(z_1,\zeta) d\eta d\zeta\cr
		&\sim\frac{e^{-\mu^*(y,T_1,T_2)}}{\sqrt{\Sigma_{zz}(T_1,T_2)}} \cr
		&\hspace{30pt}\left(1 +\left(\int_{T_1}^{T_2} \mG_1(T_1,t_1,T_2)dt_1 -\mQ(T_1,T_2)\right) \left(\frac{\partial}{\partial z} -1\right)\right)\cr
		&\hspace{30pt}\int_{z_1+\Delta z^*}^\infty\left(e^{\zeta-z_1} -\kappa^{-1}D(T_1,T_2)\right)  N\left(\frac{\zeta-\mu^*(y,T_1,T_2)+\frac12\Sigma_{zz}(T_1,T_2)-z} {\sqrt{\Sigma_{zz}(T_1,T_2)}} \right)d\zeta\Bigg|_{z=z_1}.\cr
	\end{align}
	We see the $z$-dependence drops out at this stage. Letting $V^{(i,j)}$ denote the result of applying $G_i(\cdot;\cdot)$ to the $j$th order contribution to $V_{\text{caplet}}$, we find
	$$V^{(0,0)}(y,T_1) = \Phi(d_1(y,0,T_1)) -\kappa^{-1}D(T_1,T_2)e^{-\mu^*(y,T_1,T_2)} \Phi(d_2(y,0,T_1)),$$
	with $d_1(\cdot)$ and $d_2(\cdot)$ defined by \eqref{eq:d_1} and \eqref{eq:d_2}, respectively.
	Applying the first-order operator and making use of the identity
	\begin{align}
		N(d_1(y,0,t)) -\kappa^{-1}D(T_1,T_2)e^{-\theta(y,t)} N(d_2(y,0,t) &= 0
	\end{align}
	yields at first order%
	\footnote{We use the convention here that the integration in the operator expression is applied \emph{after} the operator has been applied to the operand; likewise for the assignment of the $z$-value.}
	\begin{equation*}
		V^{(1,0)}(y,T_1) =\kappa^{-1}D(T_1,T_2)e^{-\mu^*(y,T_1,T_2)} \left(\int_{T_1}^{T_2} \mG_1(T_1,t_1,T_2)dt_1 -\Sigma_{zz}(T_1,T_2)\frac{\partial}{\partial z} \right) \Phi(d_2(y,z-z_1,T_1)) \bigg|_{z=z_1}.
	\end{equation*}
	Here the effect of the shift operators is to transform
	$$\Phi(d_2(y,0,T_1) \to \Phi(d_2(y,\pm\gamma(t_1)\Delta z_1(t_1),T_1) ,$$
	where
	\begin{equation}
		\Delta z_1(t_1):=\Sigma_{rz}(T_1,t_1) +B^+(T_1,t_1,T_2)\Sigma_{rr}(T_1,t_1). \label{eq:Delta_z_1}
	\end{equation}
	Combining this with the leading-order term (and noting that there is no first-order payoff contribution), we obtain $V_{\text{caplet}}(y,T_1)\sim V^{(0,0)}(y,T_1) +V^{(1,0)}(y,T_1)$. Taking this in turn as the payoff at $T_1$ and valuing as of time $t\in[0,T_1)$ gives rise to
	\begin{align}
		V_{\text{caplet}}(y,t)&=\iint_{\R^2} G(y,z,t;\eta,\zeta,T_1)V_{\text{caplet}}(\eta,T_1) d\eta d\zeta\cr
		&\sim V^{(0,0)}(y,t) +V^{(0,1)}(y,t) +V^{(1,0)}(y,t),
	\end{align}
 	The $\zeta$-integration is in this case trivial. In particular we have the quasi-Hull-White result
	\begin{align}\label{eq:V^(0,0)}
		V^{(0,0)}(y,t)= e^{-\mu^*(y,t,T_1)}\left(D(t,T_1)\Phi(d_1(y,0,t)) -\kappa^{-1}D(t,T_2)e^{-\theta(y,t)} \Phi(d_2(y,0,t))\right),
	\end{align}
	whence
	\begin{align}\label{eq:V^(0,0)(0,0)}
	V^{(0,0)}(0,0)= D(0,T_1)\Phi(d_1(0,0,0)) -\kappa^{-1}D(0,T_2)\Phi(d_2(0,0,0)),
	\end{align}
	Considering next the action of the first-order Green's function term on the leading-order term, we obtain
	\begin{align}\label{eq:V^(1,0)}
		V^{(1,0)}(y,t)&=-e^{-\mu^*(y,t,T_1)}\left(\int_t^{T_1}\mG_1(t,t_1,T_1)dt_1 -\frac{\Sigma_{rz}(t,T_1)}{\phi_r(t,T_1)} \frac{\partial}{\partial y}\right)\cr
		&\hspace{100pt}\left(D(t,T_1)\Phi(d_1(y,0,t)) -\kappa^{-1}D(t,T_2) e^{-\theta(y,t)}\Phi(d_2(y,0,t))\right).\qquad
	\end{align}
	Setting $y=t=0$, we find
	\begin{align}\label{eq:V^(1,0)(0,0)}
	V^{(1,0)}(0,0)&\sim-D(0,T_1)\int_0^{T_1}\mG_1(0,t_1,T_1)dt_1\, \Phi(d_1(y,0,0))\Bigg|_{y=0}\cr
	&\quad+\kappa^{-1}D(0,T_2)\Bigg(\int_0^{T_1}\mG_1(0,t_1,T_2)dt_1\,e^{-\theta(y,0)} +B^*(T_1,T_2)\Sigma_{rz}(0,T_1)\Bigg)\Phi(d_2(y,0,0))\Bigg|_{y=0}.
	\end{align}
	Considering next the action of the leading-order Green's function term on the first-order term $V^{(1,0)}(y,T_1) $, we obtain
	\begin{align}
		V^{(0,1)}(y,t)&\sim\kappa^{-1}D(t,T_2)e^{-\mu^*(y,t,T_1)-\theta(y,t)}\cr
		&\hspace{20pt}\bigg(\int_{T_1}^{T_2}R^+(y-\Delta y_1(t),t,t_1) e^{-\gamma(t_1)(B^+(T_1,t_1,T_2)\Sigma_{rr}(T_1,t_1)+\Sigma_{rz}(T_1,t_1))}\cr
		&\hspace{120pt}\Phi(d_2(y+\gamma(t_1)\phi_r(T_1,t_1)\Delta y(t,T_1),\gamma(t_1)\Delta z_1(t_1),t)) dt_1\cr
		&\hspace{30pt}-\int_{T_1}^{T_2}R^-(y+\Delta y_1(t),t,t_1) e^{\gamma(t_1)(B^+(T_1,t_1,T_2)\Sigma_{rr}(T_1,t_1)+\Sigma_{rz}(T_1,t_1))}\cr
		&\hspace{120pt}\Phi(d_2(y-\gamma(t_1)\phi_r(T_1,t_1)\Delta y(t,T_1), -\gamma(t_1)\Delta z_1(t_1),t))dt_1\cr
		&\hspace{30pt}+\bigg(\int_{T_1}^{T_2}\left(R_1^*(t_1) +e^{\frac12\gamma^2(t_1)\Sigma_{rr}(0,t_1)} B^+(T_1,t_1,T_2) \Sigma_{rr}(t,t_1)\right)dt_1 -\theta(y,t)\bigg)\cr
		&\hspace{330pt}\Phi(d_2(y,0,t))\cr
		&\hspace{30pt} -\sqrt{B^{*2}(T_1,T_2)\Sigma_{rr}(t,T_1) +\Sigma_{zz}(T_1,T_2)} N(d_2(y,0,t))\bigg), \label{eq:V^(0,1)}
	\end{align}
	with $\Oo(\epsilon)$ relative errors, where
	\begin{equation}
		\Delta y_1(t):=\frac{\Sigma_{rz}(t,T_1)}{\phi_r(t,T_1)} +B^*(T_1,T_2) \frac{\Sigma_{rr}(t,T_1)} {\phi_r(t,T_1)}. \label{eq:Delta_y_1}
	\end{equation}
	We conclude, making use of the fact that $\mu^*(0,0,T_1) =\theta(0,0)=0$, that
	\begin{align}\label{eq:V^(0,1)(0,0)}
		V^{(0,1)}(0,0)&\sim\kappa^{-1}D(0,T_2)\bigg(\int_{T_1}^{T_2} R_1^+(0,0,t_1,T_2) \Phi(d_2(\gamma(t_1) \phi_r(T_1,t_1) \Delta y(0,T_1), \gamma(t_1)\Delta z_1(t_1),0))dt_1 \cr
		&\hspace{75pt}-\int_{T_1}^{T_2} R_1^-(0,0,t_1,T_2)\Phi(d_2(-\gamma(t_1) \phi_r(T_1,t_1) \Delta y(0,T_1), -\gamma(t_1)\Delta z_1(t_1),0))dt_1 \cr
		&\hspace{75pt}+\int_{T_1}^{T_2} R_1^*(t_1)dt_1 \Phi(d_2(0,0,0))\cr
		&\hspace{75pt}  -\sqrt{B^{*2}(T_1,T_2)\Sigma_{rr}(0,T_1) +\Sigma_{zz}(T_1,T_2)} N(d_2(0,0,0)) \bigg).
 	\end{align}
	We make use here of the identity that, for $t\le T_1\le t_1\le T_2$,
	\begin{align*}
		 \Sigma_{rz}(&T_1,t_1) +B^+(T_1,t_1,T_2)\Sigma_{rr}(T_1,t_1) +\phi_r(t,t_1)\Delta y_1(t,T_1)\cr
		&=\Sigma_{rz}(T_1,t_1) +B^+(T_1,t_1,T_2)\Sigma_{rr}(t,t_1) +\phi_r(T_1,t_1)(\Sigma_{rz}(t,T_1) +B^*(T_1,t_1)\Sigma_{rr}(t,T_1))\cr
		&=\Sigma_{rz}(t,t_1) +B^+(t,t_1,T_2)\Sigma_{rr}(t,t_1) -\phi_r(t,t_1)\Delta y_2(t,t_1)
	\end{align*}
	where
	\begin{align}\label{eq:Delta_y_2}
		\Delta y_2(t,t_1)&=\frac1{\phi_r(t,t_1)}\int_{T_1}^{T_2} \left(\psi_r(t,u) -\psi_r(T_1,t_1)\right)\cr
		&\hspace{100pt}\left(\phi_r(u,t_1)\Sigma_{rr}(T_1,u)\mathbbm1_{u\le t_1} +\phi_r(T_1,u)\Sigma_{rr}(t,T_1)\mathbbm1_{u>t_1}\right)du.
	\end{align}
	We neglect this subdominant adjustment factor in our asymptotic estimate, arguing that the result would otherwise not be consistent with the price of deep in-the-money forward contracts. Combining all $V^{(i,j)}(0,0)$ terms and simplifying gives rise to \eqref{eq:PV_caplet}. This completes the proof.
\end{proof}

\section{Calibration to Caplet market data}\label{sec:Calibration_2}
The calibration of the forward curve has already been explained in Sec.~\ref{sec:Calibration}. As mentioned before, for the mean reversion, we choose a representative fixed value of $\alpha(t) = 0.15$.
The calibration of $\sigma(t)$, $y^*(t)$ and $\gamma(t)$ follows from matching the analytical implied volatility formula to the implied volatilities quoted in the market for 6-month caps. We take piece-wise parametrisations for these parameters and calibrate on the available maturities. The result of the calibration is given in Fig. \ref{fig:calibrated_params}.

\begin{figure}[H]
	\centering
	\includegraphics[width=0.32\linewidth]{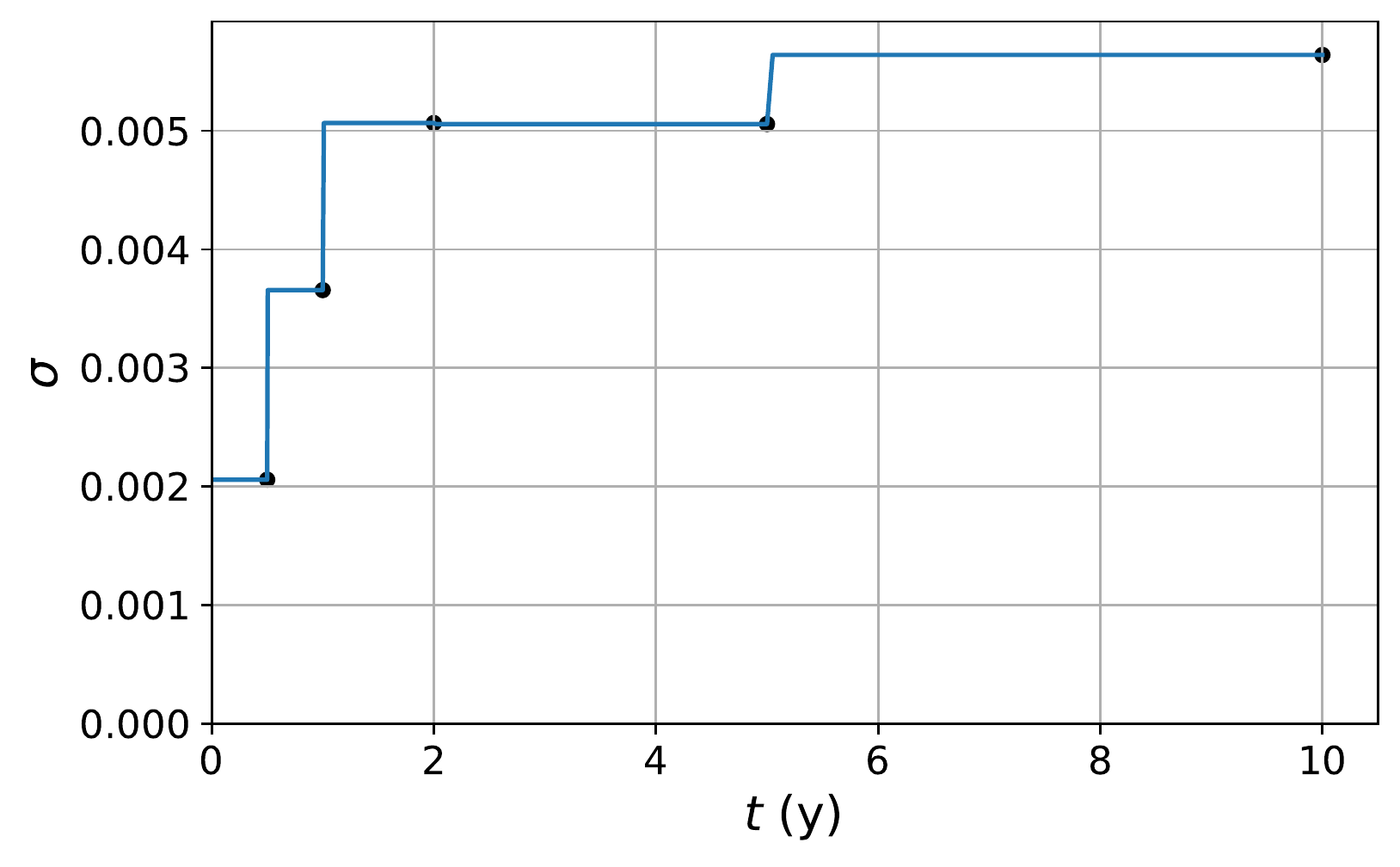}
	\includegraphics[width=0.32\linewidth]{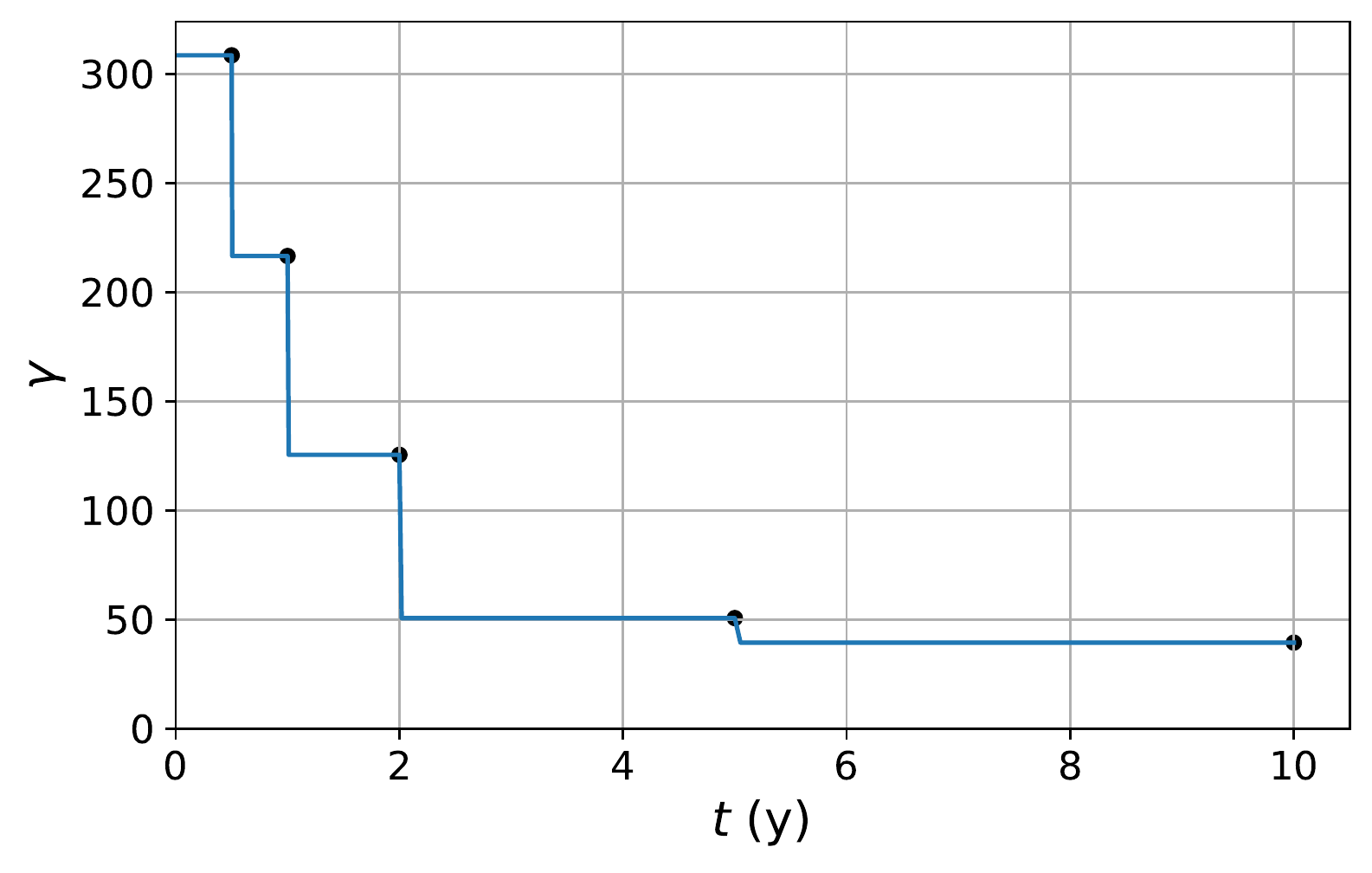}
	\includegraphics[width=0.32\linewidth]{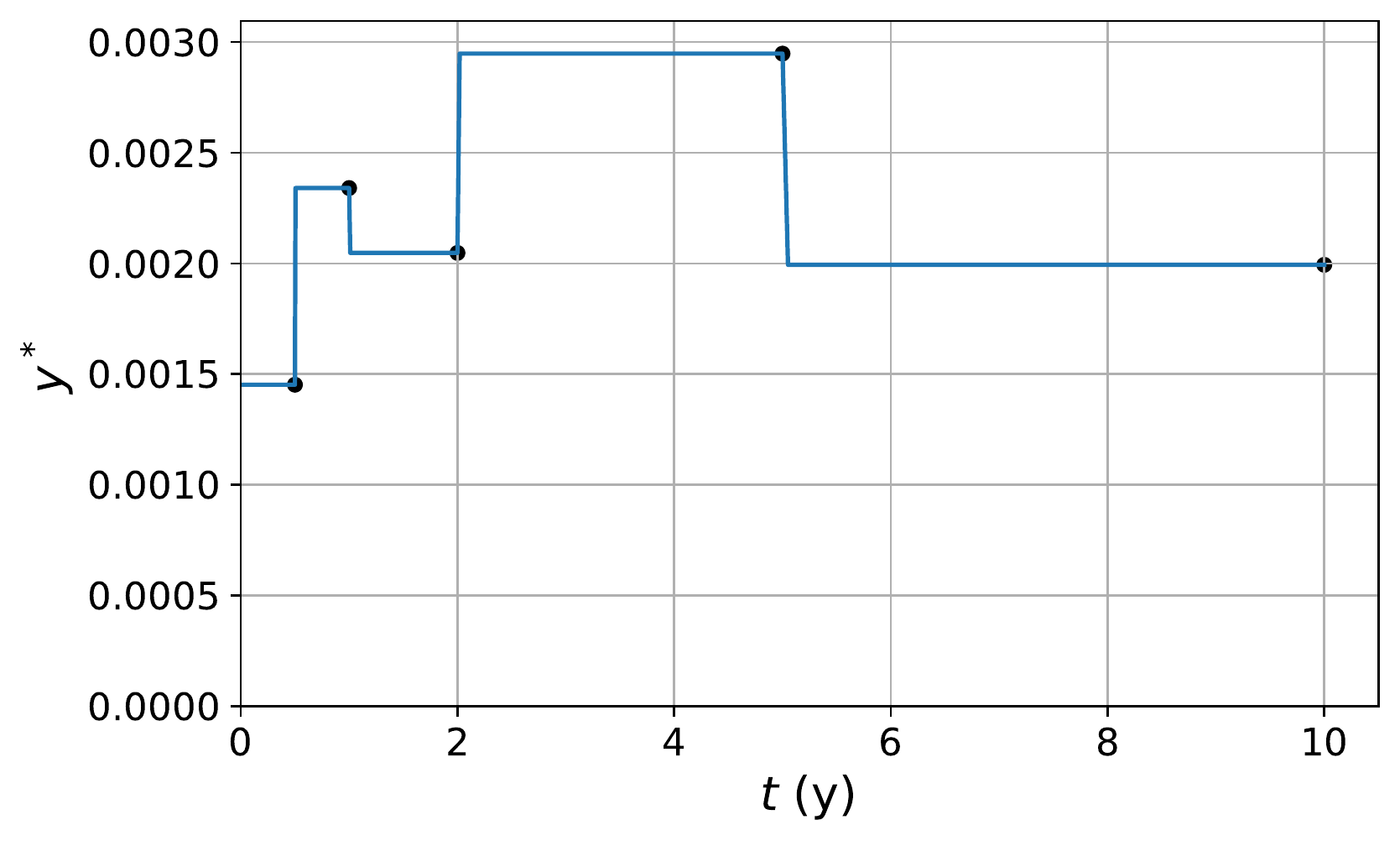}
	\caption{Piece-wise functions for $\sigma(t)$, $y^*(t)$ and $\gamma(t)$ after calibration to market data. The black dots indicate the cap maturities used for calibration.}
	\caption{Piece-wise functions for $\sigma(t)$, $y^*(t)$ and $\gamma(t)$ after calibration to market data. The black dots indicate the cap maturities used for calibration.}
	\label{fig:calibrated_params}
\end{figure}

\addcontentsline{toc}{section}{References}

\bibliography{Bibliography}

\end{document}